\documentclass[a4paper,12pt]{article}

\usepackage{amsmath,amssymb}
\usepackage{amscd}
\usepackage{mathrsfs}
\usepackage{amsthm}
\usepackage{color}

\usepackage{verbatim}

\usepackage{mathtools}
\usepackage{enumerate}

\usepackage[backref=page, colorlinks=true, linkcolor=blue, citecolor=blue, urlcolor=blue]{hyperref}

\theoremstyle{plain}

\newtheorem{theorem}{\bf Theorem}[section]
\newtheorem{lemma}[theorem]{\bf Lemma}
\newtheorem{proposition}[theorem]{\bf Proposition}
\newtheorem{corollary}[theorem]{\bf Corollary}

\theoremstyle{definition}
\newtheorem{definition}[theorem]{\bf Definition}
\newtheorem{example}[theorem]{\bf Example}
\newtheorem{remark}[theorem]{\bf Remark}

\newtheorem{question}[theorem]{\bf {Question}}

  \makeatletter
  \newcommand{\subsubsubsection}{\@startsection{paragraph}{4}{\z@}%
    {1.0\Cvs \@plus.5\Cdp \@minus.2\Cdp}%
    {.1\Cvs \@plus.3\Cdp}%
    {\reset@font\sffamily\normalsize}
  }
  \makeatother
\makeatletter

\@addtoreset{equation}{section}
\makeatother

\title{Representation Theory of Canonical Commutation Relations Arising from the Quantization of the Klein-Gordon Equation with an External Potential}

\author{Yasumichi Matsuzawa\thanks{Department of Mathematics, Faculty of Education, Shinshu University, 6-Ro, Nishi-nagano, Nagano 380-8544, Japan, e-mail: myasu@shinshu-u.ac.jp} 
	 \and Itaru Sasaki\thanks{Department of Mathematics,  Faculty of Science, Shinshu University, Matsumoto 390-8621, Japan, e-mail: isasaki@shinshu-u.ac.jp}
     \and Akito Suzuki\thanks{Department of Production Systems Engineering and Sciences, Komatsu University, Awazu Campus, Nu 1-3 Shicho-machi, Komatsu, Ishikawa 923-8511, Japan. e-mail: akito.suzuki@komatsu-u.ac.jp}}
\date{\today}

\begin{document}
\maketitle

\begin{abstract}
	We investigate the Weyl representation of the canonical commutation relations for a model describing a quantized massive scalar field under the influence of an external potential.
	The main problem is to determine whether the Weyl representation remains equivalent to or becomes inequivalent to the original one when the mass and/or potential are changed.
	This problem is reduced to the study of Schr\"{o}dinger operators.
	It turns out that the Weyl representations are inequivalent when the masses differ.
	Moreover, when the masses are the same, the transition between equivalence and inequivalence occurs at the decay rate $-3/2$ of the difference between the potentials. 
	This contrasts with the decay rate $-1$ that defines the short-range condition in scattering theory.
\end{abstract}

\section{Introduction}
We consider the external field problem, which concerns the behavior of a quantized field interacting with a classical external field.
Typical examples include the quantized Klein--Gordon field and the quantized Dirac field interacting with a classical electromagnetic field.
For the mathematical study of the external field problem, see \cite{MR3390679,MR446207,MR1219537} and the references therein.

In this paper, we focus on a model of a quantized real scalar field $\phi$ with positive mass $m>0$, affected by an external potential $V:\mathbb{R}^3\to\mathbb{R}$ \cite{MR3390679,MR356771,MR2840105}. 
The field $\phi$ obeys the equation
\begin{equation}\label{KGeq}
\left(\frac{\partial^2}{\partial t^2}-\Delta +m^2+V(x)\right)\phi(t,x)=0,\qquad t\in\mathbb{R},\ x\in\mathbb{R}^3,
\end{equation}
where $\Delta$ is the Laplacian on $\mathbb{R}^3.$ 
Here, we regard $\phi(t,x)$ as an operator-valued distribution kernel.
The model consists of an irreducible Weyl representation $\rho_{m,V}$ of the canonical commutation relations on the Boson Fock space $\mathscr{F}_\mathrm{b}(L^2(\mathbb{R}^3))$ and the second quantization of the operator $\sqrt{-\Delta+m^2+V}.$
The former represents the time-zero fields of the model, while the latter is referred to as the Hamiltonian of the model.
The aim of this paper is to investigate whether the two Weyl representations $\rho_{m_1,V_1}$ and $\rho_{m_2,V_2}$ are (unitarily) equivalent when $(m_1,V_1)\not=(m_2,V_2).$
This problem is reduced to the study of Schr\"{o}dinger operators of the form $-\Delta+m^2+V$ as a consequence of the theory of Bogoliubov transformations.
Our main results can be summarized as follows:
\begin{itemize}
	\item if $m_1\not=m_2,$ then $\rho_{m_1,V_1}$ and $\rho_{m_2,V_2}$ are not equivalent [Theorem \ref{m_1not=m_2}].
\end{itemize}
In the case where $m:=m_1=m_2$, we obtain the following results.
\begin{itemize}
	\item If $V_1-V_2$ belongs to $\mathcal{R}+L^2(\mathbb{R}^3),$ where $\mathcal{R}$ denotes the Rollnik class, then $\rho_{m,V_1}$ and $\rho_{m,V_2}$ are equivalent [Theorem \ref{T_{m,V} eq if R+L^2}], 
	\item if $V_1-V_2$ decays at a rate slower than or equal to $|x|^{-3/2}$ as $|x|\to\infty$, then $\rho_{m,V_1}$ and $\rho_{m,V_2}$ are not equivalent [Theorem \ref{T_{m,V} ineq}],
	\item if $V_2$ is of the form $V_2=V_1+\gamma/|x|^{\sigma},$ where $\gamma\not=0$ and $0<\sigma<2$ are real constants, then  $\rho_{m,V_1}$ and $\rho_{m,V_2}$ are equivalent if and only if $\sigma>3/2$ [Theorem \ref{coulomb type thm}].
\end{itemize}

The first main result generalizes the well-known theorem that $\rho_{m_1,0}$ is not equivalent to $\rho_{m_2,0}$ whenever $m_1\not=m_2$ (see e.g., \cite[Theorem 1.1]{MR3513945}, \cite[Theorem 10.13]{MR4292535} or \cite[Theorem X.46]{MR493420}).
The second main result generalizes those obtained in \cite{MR356771,MR2840105}.
In \cite{MR356771}, the case where $V_1$ is non-negative, infinitely differentiable, and has compact support, with $V_2=0,$ was considered.
In \cite{MR2840105}, the case where $V_1$ is non-negative and belongs to the Sobolev space $H^2(\mathbb{R}^3),$ with $V_2=0,$ was studied.
The last two main results show that the transition between equivalence and inequivalence occurs at the decay rate $|x|^{-3/2}$ of $V_1-V_2$, which differs from the decay rate $|x|^{-1}$ that defines the short-range condition in scattering theory.

This paper is organized as follows.
In Section \ref{Preliminaries}, we review the general theory of Weyl representations of the canonical commutation relations and the fundamental concepts of abstract free Bose field models. 
In particular, we present the result (\cite[Theorem 5.1]{MR4292535} and \cite[Theorem 1.2]{M2026classification}) that characterizes the equivalence of two Weyl representations in terms of their one-particle Hamiltonians.
In Section \ref{Equivalence relations}, we introduce two equivalence relations on the set of injective, non-negative self-adjoint operators acting in an abstract Hilbert space.
These equivalence relations are an abstraction of the method used in \cite{MR4213757} and play a key role in proving the main results of this paper.
In Section \ref{Subspaces spread by self-adjoint operators}, we introduce a tool for constructing Weyl representations.
In Section \ref{Quantization of Klein-Gordon equation with an external potential}, we explicitly define the quantized scalar field model that obeys the field equation \eqref{KGeq}, including its Weyl representation $\rho_{m,V}$, and prove the main results mentioned above.
In the Appendix, we summarize some results from operator theory that are used in this paper.


\section{Preliminaries}\label{Preliminaries}

In this section, we review the general theory of Weyl representations of the canonical commutation relations and the fundamental concepts of abstract free Bose field models. 
Abstract free Bose field models provide a general framework that includes our model.
The primary references are \cite{MR4292535,MR4812858}.

\subsection{Notation}

The symbol $\mathbb{N}$ denotes the set of integers greater than or equal to one.
We say that a real number $x$ is positive if $x>0.$

Let $\mathscr{H}$ be an infinite-dimensional, separable complex Hilbert space.
Its inner product $\langle f,g\rangle$ is anti-linear in $f\in\mathscr{H}$ and linear in $g\in\mathscr{H}.$

Let $A$ be a linear operator acting in $\mathscr{H}.$
The domain of $A$ is denoted by $\mathrm{dom}(A).$
We say that $A$ is everywhere defined if $\mathrm{dom}(A)$ coincides with $\mathscr{H}.$
We say that $A$ has a bounded inverse if it is bijective and its inverse $A^{-1}$ is bounded.
If $A$ is densely defined, its adjoint is denoted by $A^*.$
If $A$ is closable, its closure is denoted by $\bar{A}.$
The identity operator on $\mathscr{H}$ is denoted by $1.$
For any $z\in\mathbb{C},$ the scalar operator $z1$ is simply denoted by $z.$

Let $A$ be a self-adjoint operator acting in $\mathscr{H}$.
The spectrum and essential spectrum of $A$ are denoted by $\sigma(A)$ and $\sigma_\mathrm{ess}(A),$ respectively.
The spectral resolution of $A$ is expressed as $A=\int_\mathbb{R}\lambda\,\mathrm{d}E_A(\lambda)$, where $E_A(\cdot)$ is the spectral measure of $A.$
We say that $A$ is non-negative if $\sigma(A)\subset[0,\infty)$.
The symbol $\mathcal{S\!A}_+(\mathscr{H})$ denotes the set of all injective, non-negative self-adjoint operators acting in $\mathscr{H}.$ 
If $A$ is non-negative, then $A$ has a bounded inverse if and only if the infimum of $\sigma(A)$ is positive.

For a bounded operator $A,$ its operator norm is denoted by $\|A\|.$
Let $\mathcal{B}(\mathscr{H})$ denote the set of all bounded operators that are everywhere defined on $\mathscr{H}$.
For each $p\in[1,\infty),$ the symbol $\mathcal{C}_p(\mathscr{H})$ denotes the set of all Schatten $p$-class operators on $\mathscr{H}$.
The set $\mathcal{C}_1(\mathscr{H})$ coincides with the set of trace class operators, while $\mathcal{C}_2(\mathscr{H})$ coincides with the set of Hilbert--Schmidt operators.
The Hilbert--Schmidt norm of $A\in\mathcal{C}_2(\mathscr{H})$ is denoted by $\|A\|_\mathrm{HS}.$

\subsection{Boson Fock spaces}

Let $\mathscr{H}$ be an infinite-dimensional, separable complex Hilbert space.
The \textit{Boson Fock space} over $\mathscr{H}$ is defined by
\[
\mathscr{F}_\mathrm{b}(\mathscr{H}):=\bigoplus_{n=0}^\infty\otimes_\mathrm{s}^n\mathscr{H},
\]
where $\otimes_\mathrm{s}^n\mathscr{H}$ is the $n$-fold symmetric tensor product Hilbert space of $\mathscr{H}$ with $\otimes_\mathrm{s}^0\mathscr{H}:=\mathbb{C}$ and $\otimes_\mathrm{s}^1\mathscr{H}:=\mathscr{H}.$
Each vector $\Psi\in\mathscr{F}_\mathrm{b}(\mathscr{H})$ can be expressed as $\Psi=\{\Psi^{(n)}\}_{n=0}^\infty,$ where $\Psi^{(n)}\in\otimes_\mathrm{s}^n\mathscr{H}$ for all $n\geq0.$ 

The \textit{creation operator} $A^\dagger(f)$ with test vector $f\in\mathscr{H}$ is given by
\begin{align*}
	\mathrm{dom}(A^\dagger(f))&:=\left\{\Psi=\{\Psi^{(n)}\}_{n=0}^\infty\in\mathscr{F}_\mathrm{b}(\mathscr{H})\,\left|\, \sum_{n=1}^\infty\|\sqrt{n}S_n(f\otimes\Psi^{(n-1)})\|_{\otimes_\mathrm{s}^n\mathscr{H}}^2<\infty\right.\right\},\\
	(A^\dagger(f)\Psi)^{(0)}&:=0,\\
	(A^\dagger(f)\Psi)^{(n)}&:=\sqrt{n}S_n(f\otimes\Psi^{(n-1)}),\qquad n\in\mathbb{N},\ \Psi\in\mathrm{dom}(A^\dagger(f)),
\end{align*}
where $S_n$ is the symmetrization operator from $\otimes^n\mathscr{H}$ onto $\otimes_\mathrm{s}^n\mathscr{H}.$
Its adjoint $A(f):=A^\dagger(f)^*$ is called the \textit{annihilation operator} with test vector $f.$
We have $A(f)^*=A^\dagger(f).$
For each $f\in\mathscr{H},$ the operator $A(f)+A(f)^*$ is essentially self-adjoint,
and thus we can define a self-adjoint operator $\Phi_\mathrm{S}(f)$ by
\[
\Phi_\mathrm{S}(f) :=\frac{1}{\sqrt{2}}\overline{A(f)+A(f)^*}.
\]
The operator $\Phi_\mathrm{S}(f)$ is called the \textit{Segal field operator} with test vector $f.$

Let $T$ be a non-negative self-adjoint operator acting in $\mathscr{H}.$
The \textit{second quantization operator} $\mathrm{d}\Gamma_\mathrm{b}(T)$ of $T$ is defined by
\[
\mathrm{d}\Gamma_\mathrm{b}(T):=\bigoplus_{n=0}^\infty T^{(n)},
\]
where $T^{(0)}:=0$ and
\[
T^{(n)}:=\sum_{j=1}^n\underbrace{1\otimes\cdots\otimes1}_{j-1}\otimes \overset{j\text{th}}{T}\otimes\underbrace{1\otimes\cdots\otimes1}_{n-j},\qquad n\geq1.
\]
It is known that $\mathrm{d}\Gamma_\mathrm{b}(T)$ is self-adjoint and non-negative.

\subsection{Weyl representations of the canonical commutation relations}\label{subsec weyl}

Let $C$ be a conjugation on $\mathscr{H}.$
That is, $C$ is an anti-linear, norm-preserving map on $\mathscr{H}$ satisfying $C^2 = 1$.
For a subspace $\mathscr{D}$ of $\mathscr{H},$ we define
\[
\mathscr{D}_C:=\{f\in\mathscr{D}\mid Cf=f\}.
\]
Then $\mathscr{D}_C$ becomes a real inner product space by restricting the inner product on $\mathscr{H}$ to $\mathscr{D}_C.$
In particular, $\mathscr{H}_C$ is a real Hilbert space.
Let $\mathscr{V}$ be a (not necessarily dense) real subspace of $\mathscr{H}_C.$
The symbol $\mathcal{S\!A}_{C,\mathscr{V}}(\mathscr{H})$ denotes the set of all injective self-adjoint operators $T$ acting in $\mathscr{H}$ that satisfy the following two conditions:
\begin{itemize}
	\item $CT\subset TC,$
	\item $\mathscr{V}\subset \mathrm{dom}(T)\cap \mathrm{dom}(T^{-1}),$ and $T\mathscr{V}$ as well as $T^{-1}\mathscr{V}$ are dense in $\mathscr{H}_C.$
\end{itemize}

For each $T\in\mathcal{S\!A}_{C,\mathscr{V}}(\mathscr{H})$ and $f\in\mathscr{V},$ we define self-adjoint operators $\phi_T(f)$ and $\pi_T(f)$ by
\[
\phi_T(f):=\Phi_{\mathrm{S}}(T^{-1}f),\qquad \pi_T(f):=\Phi_{\mathrm{S}}(\mathrm{i}Tf).
\]
Note that they act in $\mathscr{F}_{\mathrm{b}}(\mathscr{H}).$
Then, the pair 
\[
\rho_T:=\{\mathscr{F}_{\mathrm{b}}(\mathscr{H}),\{\phi_T(f),\pi_T(f)\mid f\in\mathscr{V}\}\}
\]
is a \textit{Weyl representation of the canonical commutation relations} over $\mathscr{V}.$
That is, it satisfies
	\begin{align*}
	\mathrm{e}^{\mathrm{i}\phi_{T}(sf+tg)}&=\mathrm{e}^{\mathrm{i}s\phi_{T}(f)}\mathrm{e}^{\mathrm{i}t\phi_{T}(g)},
	\qquad \mathrm{e}^{\mathrm{i}\pi_{T}(sf+tg)}=\mathrm{e}^{\mathrm{i}s\pi_{T}(f)}\mathrm{e}^{\mathrm{i}t\pi_{T}(g)},\\
	\mathrm{e}^{\mathrm{i}\phi_{T}(f)}\mathrm{e}^{\mathrm{i}\pi_{T}(g)}
	&=\mathrm{e}^{-\mathrm{i}\langle f,g\rangle}\mathrm{e}^{\mathrm{i}\pi_{T}(g)}\mathrm{e}^{\mathrm{i}\phi_{T}(f)},
	\qquad\forall s,t\in\mathbb{R},\ f,g\in\mathscr{V}.
\end{align*}
Moreover, $\rho_T$ is \textit{irreducible}, meaning that the only operators in $\mathcal{B}(\mathscr{F}_{\mathrm{b}}(\mathscr{H}))$ that commute with all elements in $\{\mathrm{e}^{\mathrm{i}\phi_T(f)},\mathrm{e}^{\mathrm{i}\pi_T(f)}\mid f\in\mathscr{V}\}$ are scalar operators.
The proofs can be found in \cite[Lemma 4.3]{MR3513945} or \cite[Theorem 5.47]{MR4812858}.
Note that their proofs assume the density of $\mathscr{V}$, although these facts do not require this assumption.

Let $T_1,T_2\in\mathcal{S\!A}_{C,\mathscr{V}}(\mathscr{H})$ be arbitrary.
We say that $\rho_{T_1}$ is \textit{equivalent} to $\rho_{T_2}$ if there exists a unitary operator $U$ on $\mathscr{F}_{\mathrm{b}}(\mathscr{H})$ such that
\begin{equation}\label{def of equiv of reprs}
	U\phi_{T_1}(f)U^*=\phi_{T_2}(f),\qquad U\pi_{T_1}(f)U^*=\pi_{T_2}(f),\qquad \forall f\in\mathscr{V}.
\end{equation}
A \textit{transfer pair} $(J_+,J_-)$ from $T_1$ to $T_2$ with respect to $(C,\mathscr{V})$ is a pair of bounded operators $J_+,J_-\in\mathcal{B}(\mathscr{H})$ satisfying
\begin{equation*}
	J_+T_1f=T_2f,\qquad J_-T_1^{-1}f=T_2^{-1}f,\qquad\forall f\in\mathscr{V}.
\end{equation*}
We note that if a transfer pair exists, it is unique since $T_1\mathscr{V}$ and $T_1^{-1}\mathscr{V}$ are dense in $\mathscr{H}_C.$

The following theorem, which is a consequence of the theory of Bogoliubov transformations, provides a necessary and sufficient condition for $\rho_{T_1}$ to be equivalent to $\rho_{T_2}$.

\begin{theorem}[{\cite[Theorem 5.1]{MR4292535}} and {\cite[Theorem 1.2]{M2026classification}}]\label{eq or ineq weyl repr}
	Let $T_1,T_2\in\mathcal{S\!A}_{C,\mathscr{V}}(\mathscr{H}).$
	Then $\rho_{T_1}$ is equivalent to $\rho_{T_2}$
	if and only if a transfer pair $(J_+,J_-)$ from $T_1$ to $T_2$ with respect to $(C,\mathscr{V})$ exists and the difference $J_+-J_-$ is Hilbert--Schmidt.
\end{theorem}

Thus, to determine whether $\rho_{T_1}$ is equivalent to $\rho_{T_2},$ it is sufficient to investigate the relationship between $T_1$ and $T_2.$

\subsection{Abstract free Bose field models}

Recall that the symbol $\mathcal{S\!A}_+(\mathscr{H})$ denotes the set of all injective, non-negative self-adjoint operators acting in $\mathscr{H}.$ 
For each $T\in\mathcal{S\!A}_+(\mathscr{H})$ with $T^{1/2}\in\mathcal{S\!A}_{C,\mathscr{V}}(\mathscr{H}),$ the triplet
\[
\mathbb{M}_T:=\big\{\mathscr{F}_{\mathrm{b}}(\mathscr{H}),\mathrm{d}\Gamma_\mathrm{b}(T),\{\phi_{T^{1/2}}(f),\pi_{T^{1/2}}(f)\mid f\in\mathscr{V}\}\big\}
\]
is called an \textit{abstract free Bose field model} with \textit{one-particle Hamiltonian} $T$ (see \cite[Example 10.2]{MR4292535} or \cite[Section 5.18 and p. 844]{MR4812858}).
In this setting, the self-adjoint operator $\mathrm{d}\Gamma_{\mathrm{b}}(T)$ is called the \textit{Hamiltonian} of the model, and the operator-valued functionals
\[
f\mapsto\phi_{T^{1/2}}(f),\qquad f\mapsto\pi_{T^{1/2}}(f)
\]
represent the \textit{time-zero fields} of the model.
Then the \textit{time-$t$ fields} are given by
\begin{align*}
	\phi_{T^{1/2}}(t,f)&:=\mathrm{e}^{\mathrm{i}t\mathrm{d}\Gamma_{\mathrm{b}}(T)}\phi_{T^{1/2}}(f)\mathrm{e}^{-\mathrm{i}t\mathrm{d}\Gamma_{\mathrm{b}}(T)}
	=\Phi_{\mathrm{S}}(\mathrm{e}^{\mathrm{i}tT}T^{-1/2}f),\\
	\pi_{T^{1/2}}(t,f)&:=\mathrm{e}^{\mathrm{i}t\mathrm{d}\Gamma_{\mathrm{b}}(T)}\pi_{T^{1/2}}(f)\mathrm{e}^{-\mathrm{i}t\mathrm{d}\Gamma_{\mathrm{b}}(T)}
	=\Phi_{\mathrm{S}}(\mathrm{i}\mathrm{e}^{\mathrm{i}tT}T^{1/2}f),\qquad t\in\mathbb{R},\ f\in\mathscr{V}.
\end{align*}
They satisfy the following differential equations:

\begin{proposition}\label{AFBFM eq}
	For any $\Psi\in\mathrm{dom}(\mathrm{d}\Gamma_{\mathrm{b}}(1)^{1/2})$ and $f\in\mathscr{V}\cap\mathrm{dom}(T^{2})$ with $T^2f\in\mathscr{V},$
	we have
	\[
	\frac{\mathrm{d}}{\mathrm{d}t}\phi_{T^{1/2}}(t,f)\Psi=\pi_{T^{1/2}}(t,f)\Psi,\qquad \frac{\mathrm{d}^2}{\mathrm{d}t^2}\phi_{T^{1/2}}(t,f)\Psi+\phi_{T^{1/2}}(t,T^2f)\Psi=0.
	\]
\end{proposition}

\begin{proof}
	This follows from \cite[Corollary 6.8 (ii)]{MR4292535} or \cite[Theorem 5.33]{MR4812858}.
\end{proof}

Thus, the abstract free Bose field model is an abstraction of a quantum field model obeying the Klein-Gordon equation \eqref{KGeq}.




\section{Equivalence relations}\label{Equivalence relations}

Let $\mathscr{H}$ be an infinite-dimensional, separable complex Hilbert space.
In this section, we introduce two equivalence relations on $\mathcal{S\!A}_+(\mathscr{H}),$ and prove that they are closed under taking powers of operators.


\subsection{Weak equivalence relation}

Let $S,T\in\mathcal{S\!A}_+(\mathscr{H})$ be arbitrary.
We write $S\preceq T$ if both $\mathrm{dom}(T^{1/2})\subset\mathrm{dom}(S^{1/2})$ and 
\[
\|S^{1/2}\psi\|\leq\|T^{1/2}\psi\|,\qquad\forall \psi\in\mathrm{dom}(T^{1/2})
\]
hold.

\begin{definition}
	Let $S,T\in\mathcal{S\!A}_+(\mathscr{H})$ be arbitrary.
	 We say that $S$ and $T$ are \textit{weakly equivalent} if there exist positive real numbers $c_1,c_2>0$ such that
		\begin{equation}\label{c1c2 def}
		c_1 T\preceq S\preceq c_2 T.
		\end{equation}
		In this case, we write $S\overset{\mathrm{w}}{\sim} T.$
\end{definition}

\begin{proposition}
	The above binary relation $\,\overset{\mathrm{w}}{\sim}$ is an equivalence relation on $\mathcal{S\!A}_+(\mathscr{H}).$
\end{proposition}

\begin{proof}
	Reflexivity: Take $T\in\mathcal{S\!A}_+(\mathscr{H}).$
	Then we have $T\preceq T\preceq T,$ and thus $T\overset{\mathrm{w}}{\sim} T.$
	
	Symmetry: Take $S,T\in\mathcal{S\!A}_+(\mathscr{H}),$ and suppose that $S\overset{\mathrm{w}}{\sim} T.$
	Since the inequality \eqref{c1c2 def} holds true for some $c_1,c_2>0,$ we have
	\[
	c_2^{-1}S\preceq T\preceq c_1^{-1}S.
	\]
	Hence $T\overset{\mathrm{w}}{\sim} S.$
	
	Transitivity: Take $R,S,T\in\mathcal{S\!A}_+(\mathscr{H}),$ and suppose that $R\overset{\mathrm{w}}{\sim} S$ and $S\overset{\mathrm{w}}{\sim} T.$
	Then we have
	\[
	b_1 S\preceq R\preceq b_2 S\quad\textrm{and}\quad c_1 T\preceq S\preceq c_2 T
	\]
	for some $b_1,b_2,c_1,c_2>0.$
	Thus it holds that
	\[
	b_1c_1 T\preceq R\preceq b_2c_2 T,
	\]
	which implies that $R\overset{\mathrm{w}}{\sim} T.$
	This finishes the proof.
\end{proof}

\begin{proposition}\label{weq iff}
	Let $S,T\in\mathcal{S\!A}_+(\mathscr{H})$ be arbitrary.
	Then $S\overset{\mathrm{w}}{\sim} T$ if and only if the following three conditions hold:
	\begin{enumerate}[(i)]
		\item $\mathrm{dom}(S^{1/2})=\mathrm{dom}(T^{1/2}),$
		\item $\mathrm{dom}(S^{-1/2})=\mathrm{dom}(T^{-1/2}),$
		\item $S^{1/2}T^{-1/2}$ and $S^{-1/2}T^{1/2}$ are bounded.
	\end{enumerate}
In this case, the following assertions hold true:
\begin{enumerate}[(1)]
	\item $T^{1/2}S^{-1/2}$ and $T^{-1/2}S^{1/2}$ are bounded,
	\item $(S^{1/2}T^{-1/2})^*=\overline{T^{-1/2}S^{1/2}}$ and $(S^{-1/2}T^{1/2})^*=\overline{T^{1/2}S^{-1/2}}$ hold,
	\item $\overline{S^{1/2}T^{-1/2}}$ has a bounded inverse with $\left(\overline{S^{1/2}T^{-1/2}}\right)^{-1}=\overline{T^{1/2}S^{-1/2}},$
	\item $\overline{S^{-1/2}T^{1/2}}$ has a bounded inverse with $\left(\overline{S^{-1/2}T^{1/2}}\right)^{-1}=\overline{T^{-1/2}S^{1/2}}.$
\end{enumerate}
\end{proposition}

\begin{proof}
	We first prove the ``only if'' part.
	Suppose that $S\overset{\mathrm{w}}{\sim} T.$
	Then, by definition, there exist positive real numbers $c_1,c_2>0$ such that the inequality \eqref{c1c2 def} holds.
	This, together with Lemma \ref{form op ine inverse}, implies that 
	\begin{equation}\label{weq iff op ine inverse}
		c_2^{-1}T^{-1}\preceq S^{-1}\preceq c_1^{-1}T^{-1}.
	\end{equation}
The inequality \eqref{c1c2 def} tells us that $\mathrm{dom}(S^{1/2})$ coincides with $\mathrm{dom}(T^{1/2})$ and
	\[
	\|S^{1/2}\psi\|\leq c_2^{1/2}\|T^{1/2}\psi\|,\qquad\forall \psi\in\mathrm{dom}(T^{1/2}).
	\]
	Replacing $\psi$ with $T^{-1/2}\psi,$ we obtain
	\[
	\|S^{1/2}T^{-1/2}\psi\|\leq c_2^{1/2}\|\psi\|,\qquad\forall\psi\in\mathrm{dom}(T^{-1/2}),
	\]
	and thus $S^{1/2}T^{-1/2}$ is bounded.
	Similarly, it follows from the inequality \eqref{weq iff op ine inverse} that $S^{-1/2}T^{1/2}$ is bounded.
	
	We next prove the ``if'' part.
	Suppose that (i), (ii) and (iii) hold true.
	Then we have
	\[
	\|S^{1/2}\psi\|\leq\|S^{1/2}T^{-1/2}\|\|T^{1/2}\psi\|,\qquad\forall\psi\in\mathrm{dom}(T^{1/2}),
	\]
	whence $S\preceq\|S^{1/2}T^{-1/2}\|^2T$ holds.
	Since $(S^{-1/2}T^{1/2})^*\supset T^{1/2}S^{-1/2}$,
	the operator $T^{1/2}S^{-1/2}$ is bounded.
	By the same argument as above, we get $\|S^{-1/2}T^{1/2}\|^{-2}T\preceq S.$
	Hence $S\overset{\mathrm{w}}{\sim} T.$
	
	Finally, we prove the rest of the proposition.
	Suppose that $S\overset{\mathrm{w}}{\sim} T.$
	Then $T\overset{\mathrm{w}}{\sim} S,$ and thus (1) follows.
	To prove (2), let $\psi\in\mathrm{dom}(T^{-1/2}),\phi\in\mathrm{dom}(S^{1/2})$ be arbitrary.
	Then
	\[
	\langle\psi,(S^{1/2}T^{-1/2})^*\phi\rangle
	=\langle S^{1/2}T^{-1/2}\psi,\phi\rangle
	=\langle \psi,T^{-1/2}S^{1/2}\phi\rangle,
	\]
	and thus we obtain $(S^{1/2}T^{-1/2})^*=\overline{T^{-1/2}S^{1/2}}.$
	Similarly, we get $(S^{-1/2}T^{1/2})^*=\overline{T^{1/2}S^{-1/2}}$.
	Hence (2) follows.
	Since
	\[
\overline{S^{1/2}T^{-1/2}}\cdot \overline{T^{1/2}S^{-1/2}}\psi=\psi
=\overline{T^{1/2}S^{-1/2}}\cdot \overline{S^{1/2}T^{-1/2}}\psi,\qquad\forall\psi\in\mathrm{dom}(S^{-1/2}),
		\] 
	we obtain (3). 
	By the same argument as above, we get (4).
	This completes the proof.
\end{proof}

\begin{proposition}\label{weak eq rel p}
	Let $S,T\in\mathcal{S\!A}_+(\mathscr{H})$ satisfy $S\overset{\mathrm{w}}{\sim} T,$ and let $p\in(0,1].$
	Then $S^p\overset{\mathrm{w}}{\sim} T^p$ and $S^{-p}\overset{\mathrm{w}}{\sim} T^{-p}$ hold.
\end{proposition}

\begin{proof}
	This follows from Lemma \ref{Heinz ine} and Lemma \ref{form op ine inverse}.
\end{proof}

The following proposition gives a necessary and sufficient condition for two self-adjoint operators with bounded inverses to be weakly equivalent.
 
\begin{proposition}\label{we suf cond for massive ops}
	Let $S,T\in\mathcal{S\!A}_+(\mathscr{H})$ be arbitrary.
	Suppose that $S$ and $T$ have bounded inverses.
	Then $S\overset{\mathrm{w}}{\sim} T$ if and only if $\mathrm{dom}(S^{1/2})=\mathrm{dom}(T^{1/2}).$
\end{proposition}

\begin{proof}
	The ``only if'' part follows from Proposition \ref{weq iff}.
	To prove the ``if'' part, suppose that $\mathrm{dom}(S^{1/2})$ coincides with $\mathrm{dom}(T^{1/2}).$
	By assumption, $S^{-1/2}$ and $T^{-1/2}$ belong to $\mathcal{B}(\mathscr{H}).$
	Since $S^{1/2}T^{-1/2}$ is a closed operator defined on the whole Hilbert space $\mathscr{H},$ the closed graph theorem implies that $S^{1/2}T^{-1/2}$ is in $\mathcal{B}(\mathscr{H}).$
	Similarly, $T^{1/2}S^{-1/2}$ belongs to $\mathcal{B}(\mathscr{H}).$ 
	Taking its adjoint, we see that $S^{-1/2}T^{1/2}$ is bounded.
	By Proposition \ref{weq iff}, we conclude that $S\overset{\mathrm{w}}{\sim} T.$
\end{proof}

\subsection{Equivalence relation}

\begin{definition}\label{def of equiv rel on sa}
	Let $S,T\in\mathcal{S\!A}_+(\mathscr{H})$ be arbitrary.
	We say that $S$ and $T$ are \textit{equivalent} if $S\overset{\mathrm{w}}{\sim} T$ and the operator
	\[
	K_{S,T}:=\overline{T^{-1/2}S^{1/2}}\cdot\overline{S^{1/2}T^{-1/2}}-1
	\] 
	is Hilbert--Schmidt.
	In this case, we write $S\sim T.$ 
\end{definition}

\begin{remark}
	Let $S,T\in\mathcal{S\!A}_+(\mathscr{H})$ satisfy $S\overset{\mathrm{w}}{\sim} T.$
	Since $\overline{S^{-1/2}T^{1/2}}$ has a bounded inverse, $K_{S,T}$ is Hilbert--Schmidt if and only if
	\[
	\overline{S^{1/2}T^{-1/2}}-\overline{S^{-1/2}T^{1/2}}=\overline{S^{-1/2}T^{1/2}}K_{S,T}
	\]
	is Hilbert--Schmidt.
	The latter condition is the same as the second condition in the definition of the equivalence relation given in \cite[Remark 5.4]{MR3513945}.
	In that paper, it is assumed that $S^{-1/2}$ and $T^{-1/2}$ are bounded, but we do not make this assumption.
\end{remark}

\begin{proposition}
	The binary relation $\sim$ defined in Definition \ref{def of equiv rel on sa} is an equivalence relation on $\mathcal{S\!A}_+(\mathscr{H}).$
\end{proposition}

\begin{proof}
	Reflexivity: Take $T\in\mathcal{S\!A}_+(\mathscr{H}).$
	Then we have $K_{T,T}=0,$ and thus $T\sim T.$
	
	Symmetry: Take $S,T\in\mathcal{S\!A}_+(\mathscr{H}),$ and suppose that $S\sim T.$
	Since $K_{S,T}$ is Hilbert--Schmidt, so is
	\[
	K_{T,S}=-\overline{S^{-1/2}T^{1/2}}K_{S,T}\overline{T^{1/2}S^{-1/2}}.
	\]
	Thus $T\sim S.$
	
	Transitivity: Take $R,S,T\in\mathcal{S\!A}_+(\mathscr{H}),$ and suppose that $R\sim S$ and $S\sim T.$
	Then it holds that
	\[
	K_{R,T}=\overline{T^{-1/2}S^{1/2}}K_{R,S}\overline{S^{1/2}T^{-1/2}}+K_{S,T}
	\]
	is Hilbert--Schmidt.
	Thus $R\sim T.$
	This finishes the proof.
\end{proof}

The following proposition gives a sufficient condition for two self-adjoint operators to be equivalent.

\begin{proposition}\label{def of L}
	Let $S,T\in\mathcal{S\!A}_+(\mathscr{H})$ satisfy $S\overset{\mathrm{w}}{\sim} T$.
	If the operator 
	\[
	L_{S,T}:=\overline{S^{1/2}T^{-1/2}}-1
	\]
	is Hilbert--Schmidt, then $S\sim T.$ 
\end{proposition}

\begin{proof}
	We have
	\[
	K_{S,T}=\overline{T^{-1/2}S^{1/2}}\left(\overline{S^{1/2}T^{-1/2}}-1\right)+\overline{T^{-1/2}S^{1/2}}-1
	=\overline{T^{-1/2}S^{1/2}}	L_{S,T}+L_{S,T}^*.
	\]
	The right-hand side is Hilbert--Schmidt, and thus the proposition follows.
\end{proof}

\begin{question}\label{rem of def of L}
	Is the converse of Proposition \ref{def of L} true? 
	That is, if $S,T\in\mathcal{S\!A}_+(\mathscr{H})$ satisfy $S\sim T$, does it follow that $L_{S,T}$ is Hilbert--Schmidt?
\end{question}

The converse of Proposition \ref{def of L} holds after taking powers of the operators.

\begin{theorem}\label{main thm ver2}
	Let $S,T\in\mathcal{S\!A}_+(\mathscr{H})$ satisfy $S\sim T,$ and let $p\in(0,1).$
	Then $L_{S^p,T^p}$ is Hilbert--Schmidt.
\end{theorem}

To prove the theorem, we need four lemmas.

\begin{lemma}\label{diff of two resolvent}
	Let $S,T\in\mathcal{S\!A}_+(\mathscr{H})$ satisfy $S\overset{\mathrm{w}}{\sim} T$, and let $z\in\mathbb{C}$ be such that $z\not\in\sigma(S)\cup\sigma(T)$.
	Then, we have
	\[
	(S-z)^{-1}-(T-z)^{-1}=-S^{1/2}(S-z)^{-1}\overline{S^{-1/2}T^{1/2}}K_{S,T}T^{1/2}(T-z)^{-1}.
	\]
\end{lemma}

\begin{proof}
	For any $\psi,\phi\in\mathscr{H}$, we have
	\begin{align*}
		&\langle\psi,(S-z)^{-1}\phi\rangle-\langle\psi,(T-z)^{-1}\phi\rangle\\
		&=\langle(S-\bar{z})^{-1}\psi,(T-z)(T-z)^{-1}\phi\rangle-\langle(S-\bar{z})(S-\bar{z})^{-1}\psi,(T-z)^{-1}\phi\rangle\\
		&=\langle(S-\bar{z})^{-1}\psi,T(T-z)^{-1}\phi\rangle-\langle S(S-\bar{z})^{-1}\psi,(T-z)^{-1}\phi\rangle\\
		&=\langle T^{1/2}(S-\bar{z})^{-1}\psi,T^{1/2}(T-z)^{-1}\phi\rangle\\
		&\qquad\qquad\qquad\qquad-\langle S^{1/2}T^{-1/2}\cdot T^{1/2}(S-\bar{z})^{-1}\psi,S^{1/2}T^{-1/2}\cdot T^{1/2}(T-z)^{-1}\phi\rangle\\
		&=-\langle T^{1/2}(S-\bar{z})^{-1}\psi,K_{S,T}T^{1/2}(T-z)^{-1}\phi\rangle.
	\end{align*}
	Thus we obtain the desired result.
\end{proof}

\begin{lemma}\label{int rep of K}
	Let $S,T\in\mathcal{S\!A}_+(\mathscr{H})$ satisfy $S\overset{\mathrm{w}}{\sim} T,$ and let $p\in(0,1).$
	Define a $\mathcal{B}(\mathscr{H})$-valued function $F_{S,T}^{(p)}$ by
	\[
	F_{S,T}^{(p)}(t):=S^{1/2}(S+t)^{-1}\overline{S^{-1/2}T^{1/2}}K_{S,T}T^{(1-p)/2}(T+t)^{-1},\qquad t>0.
	\]
	Then, we have
	\[
	\langle \psi,L_{S^p,T^p}\phi\rangle = \frac{\sin{(\pi p/2)}}{\pi}\int_0^\infty\langle\psi,F_{S,T}^{(p)}(t)\phi\rangle t^{p/2}\,\mathrm{d}t,\qquad\forall\psi,\phi\in\mathscr{H}.
	\]
\end{lemma}

\begin{proof}
	It follows from Lemma \ref{gen int formula1} that
	\begin{align*}
		&\int_0^\infty|\langle\psi,F_{S,T}^{(p)}(t)\phi\rangle| t^{p/2}\,\mathrm{d}t\\
		&\leq \left\|\overline{S^{-1/2}T^{1/2}}K_{S,T}\right\|\int_0^\infty\left\|S^{1/2}(S+t)^{-1}\psi\right\| \|T^{(1-p)/2}(T+t)^{-1}\phi\| t^{p/2}\,\mathrm{d}t\\
		&\leq \left\|\overline{S^{-1/2}T^{1/2}}K_{S,T}\right\|\left[\int_0^\infty\left\|S^{1/2}(S+t)^{-1}\psi\right\|^2\,\mathrm{d}t\right]^{1/2}\\
		&\qquad\qquad\qquad\qquad\qquad\qquad\qquad\qquad\times\left[\int_0^\infty \|T^{(1-p)/2}(T+t)^{-1}\phi\|^2 t^p\,\mathrm{d}t\right]^{1/2}\\
		&= (C^{(p)})^{1/2}\left\|\overline{S^{-1/2}T^{1/2}}K_{S,T}\right\|\|\psi\|\|\phi\|
	\end{align*}
	for any $\psi,\phi\in\mathscr{H}.$
	Thus there exists a unique bounded operator $A$ such that
	\[
	\langle\psi, A\phi\rangle =  \frac{\sin{(\pi p/2)}}{\pi}\int_0^\infty\langle\psi,F_{S,T}^{(p)}(t)\phi\rangle t^{p/2}\,\mathrm{d}t,\qquad \forall\psi,\phi\in\mathscr{H}.
	\]
	It is now enough to show that $A=L_{S^p,T^p}.$
	For this, let $\psi,\phi\in\mathrm{dom}(T^{p/2})\cap\mathrm{dom}(T^{-p/2})$ be arbitrary.
	Then, by Lemma \ref{integral rep power p ver2}, we have
	\begin{align*}
		&\langle\psi, L_{S^p,T^p}\phi\rangle 
		= \langle\psi,S^{p/2}T^{-p/2}\phi\rangle-\langle\psi,T^{p/2}T^{-p/2}\phi\rangle\\
		&=\frac{\sin{(\pi p/2)}}{\pi}\int_0^\infty \frac{\langle \psi, S(S+t)^{-1}T^{-p/2}\phi\rangle}{t^{1-p/2}}\,\mathrm{d}t\\
		&\qquad\qquad\qquad\qquad-\frac{\sin{(\pi p/2)}}{\pi}\int_0^\infty \frac{\langle \psi, T(T+t)^{-1}T^{-p/2}\phi\rangle}{t^{1-p/2}}\,\mathrm{d}t.
	\end{align*}
	Lemma \ref{diff of two resolvent} tells us that the numerator of the integrand of the right-hand side can be computed as
	\begin{align*}
		&\langle \psi, S(S+t)^{-1}T^{-p/2}\phi\rangle-\langle \psi, T(T+t)^{-1}T^{-p/2}\phi\rangle\\
		&=-t\langle \psi, (S+t)^{-1}T^{-p/2}\phi\rangle+t\langle \psi, (T+t)^{-1}T^{-p/2}\phi\rangle\\
		&=t\langle \psi,F_{S,T}^{(p)}(t)\phi\rangle.
	\end{align*}
	Hence we have
	\[
	\langle\psi, L_{S^p,T^p}\phi\rangle =\frac{\sin{(\pi p/2)}}{\pi}\int_0^\infty\langle \psi,F_{S,T}^{(p)}(t)\phi\rangle t^{p/2}\,\mathrm{d}t=\langle\psi, A\phi\rangle.
	\]
	Since $\psi,\phi\in\mathrm{dom}(T^{p/2})\cap\mathrm{dom}(T^{-p/2})$ are arbitrary, we get $L_{S^p,T^p}=A$.
	This finishes the proof.
\end{proof}

\begin{lemma}\label{main thm trace class ver}
	Let $S,T\in\mathcal{S\!A}_+(\mathscr{H})$ satisfy $S\sim T,$ and let $p\in(0,1).$
	If $K_{S,T}$ is of trace class, then $L_{S^p,T^p}$ is Hilbert--Schmidt with
	\[
	\|L_{S^p,T^p}\|_{\mathrm{HS}}\leq\frac{\sin{(\pi p/2)}}{\pi}(C^{(p)})^{1/2}\|S^{-1/2}T^{1/2}\|\|K_{S,T}\|_\mathrm{HS},
	\]
	where the constant $C^{(p)}$ is defined in Lemma \ref{gen int formula1}.
\end{lemma}

\begin{proof}
	Throughout the proof, we denote $F_{S,T}^{(p)},$ $K_{S,T}$ and $L_{S^p,T^p}$ by $F,$ $K$ and $L,$ respectively. 
	Let $\{e_n\}_{n=1}^\infty$ be an orthonormal basis of $\mathscr{H}.$
	It follows from Lemma \ref{int rep of K} that
	\[
	\sum_{n=1}^\infty\|Le_n\|^2=\sum_{n=1}^\infty\langle Le_n,Le_n \rangle 
	=\frac{\sin^2{(\pi p/2)}}{\pi^2}\sum_{n=1}^\infty\int_0^\infty\int_0^\infty\langle F(s)e_n, F(t)e_n\rangle s^{p/2}t^{p/2}\,\mathrm{d}s\,\mathrm{d}t.
	\]
	We write $K$ as
	\[
	K=\sum_{j=1}^\infty\lambda_j|f_j\rangle\langle f_j|,
	\]
	where $\{\lambda_j\}_{j=1}^\infty$ is an absolutely summable sequence of real numbers and $\{f_j\}_{j=1}^\infty$ is an orthonormal basis of $\mathscr{H}.$
	Then we have
	\begin{align*}
		&\frac{\pi^2}{\sin^2{(\pi p/2)}}\sum_{n=1}^\infty\|Le_n\|^2\\
		&=\sum_{n=1}^\infty\int_0^\infty\int_0^\infty\sum_{j=1}^\infty\sum_{k=1}^\infty\lambda_j\lambda_k\langle T^{(1-p)/2}(T+s)^{-1}e_n, f_j\rangle\langle  f_k,T^{(1-p)/2}(T+t)^{-1}e_n\rangle\\ 
		&\quad\times\langle S^{1/2}(S+s)^{-1}\overline{S^{-1/2}T^{1/2}}f_j,S^{1/2}(S+t)^{-1}\overline{S^{-1/2}T^{1/2}}f_k\rangle s^{p/2}t^{p/2}\,\mathrm{d}s\,\mathrm{d}t.
	\end{align*}
	To use Fubini's theorem, we check the absolute integrability of the above integral.
	The inequality
	\begin{align*}
		&\sum_{n=1}^\infty|\langle T^{(1-p)/2}(T+s)^{-1}e_n, f_j\rangle\langle  f_k,T^{(1-p)/2}(T+t)^{-1}e_n\rangle|\\
		&\qquad\qquad\qquad\qquad\leq \|T^{(1-p)/2}(T+s)^{-1}f_j\|\|T^{(1-p)/2}(T+t)^{-1}f_k\|,
	\end{align*}
	and Lemma \ref{gen int formula1} yield that
	\begin{align*}
		&\sum_{n=1}^\infty\int_0^\infty\int_0^\infty\sum_{j=1}^\infty\sum_{k=1}^\infty|\lambda_j||\lambda_k||\langle T^{(1-p)/2}(T+s)^{-1}e_n, f_j\rangle||\langle  f_k,T^{(1-p)/2}(T+t)^{-1}e_n\rangle|\\ 
		&\quad\times |\langle S^{1/2}(S+s)^{-1}\overline{S^{-1/2}T^{1/2}}f_j,S^{1/2}(S+t)^{-1}\overline{S^{-1/2}T^{1/2}}f_k\rangle| s^{p/2}t^{p/2}\,\mathrm{d}s\,\mathrm{d}t\\
		&\leq\int_0^\infty\int_0^\infty\sum_{j=1}^\infty\sum_{k=1}^\infty|\lambda_j||\lambda_k|\|T^{(1-p)/2}(T+s)^{-1}f_j\|\|T^{(1-p)/2}(T+t)^{-1}f_k\|\\ 
		&\qquad\qquad\times \|S^{1/2}(S+s)^{-1}\overline{S^{-1/2}T^{1/2}}f_j\|\|S^{1/2}(S+t)^{-1}\overline{S^{-1/2}T^{1/2}}f_k\| s^{p/2}t^{p/2}\,\mathrm{d}s\,\mathrm{d}t\\
		&\leq C^{(p)}\left\|\overline{S^{-1/2}T^{1/2}}\right\|^2\sum_{j=1}^\infty\sum_{k=1}^\infty|\lambda_j||\lambda_k|
		<\infty.
	\end{align*}
	Hence, by Fubini's theorem, we get
	\begin{align*}
		&\frac{\pi^2}{\sin^2{(\pi p/2)}}\sum_{n=1}^\infty\|Le_n\|^2\\
		&=\int_0^\infty\int_0^\infty\sum_{j=1}^\infty\sum_{k=1}^\infty\lambda_j\lambda_k\sum_{n=1}^\infty\langle T^{(1-p)/2}(T+s)^{-1}e_n, f_j\rangle\langle  f_k,T^{(1-p)/2}(T+t)^{-1}e_n\rangle\\ 
		&\quad\times\langle S^{1/2}(S+s)^{-1}\overline{S^{-1/2}T^{1/2}}f_j,S^{1/2}(S+t)^{-1}\overline{S^{-1/2}T^{1/2}}f_k\rangle s^{p/2}t^{p/2}\,\mathrm{d}s\,\mathrm{d}t\\
		&\leq\int_0^\infty\int_0^\infty\sum_{j=1}^\infty\sum_{k=1}^\infty\frac{|\lambda_j|^2+|\lambda_k|^2}{2}|\langle T^{(1-p)/2}(T+t)^{-1}f_k, T^{(1-p)/2}(T+s)^{-1}f_j\rangle|\\ 
		&\quad\times |\langle S^{1/2}(S+s)^{-1}\overline{S^{-1/2}T^{1/2}}f_j,S^{1/2}(S+t)^{-1}\overline{S^{-1/2}T^{1/2}}f_k\rangle| s^{p/2}t^{p/2}\,\mathrm{d}s\,\mathrm{d}t.
	\end{align*}
	By Lemma \ref{gen int formula1}, we have
	\begin{align*}
		&\int_0^\infty\int_0^\infty\sum_{j=1}^\infty\sum_{k=1}^\infty|\lambda_j|^2|\langle T^{(1-p)/2}(T+t)^{-1}f_k, T^{(1-p)/2}(T+s)^{-1}f_j\rangle|\\ 
		&\qquad\times |\langle S^{1/2}(S+s)^{-1}\overline{S^{-1/2}T^{1/2}}f_j,S^{1/2}(S+t)^{-1}\overline{S^{-1/2}T^{1/2}}f_k\rangle| s^{p/2}t^{p/2}\,\mathrm{d}s\,\mathrm{d}t\\
		&\leq\int_0^\infty\int_0^\infty\sum_{j=1}^\infty|\lambda_j|^2\|T^{(1-p)/2}(T+t)^{-1}T^{(1-p)/2}(T+s)^{-1}f_j\|\\ 
		&\qquad\times \left\|\left(S^{-1/2}T^{1/2}\right)^*S^{1/2}(S+t)^{-1}S^{1/2}(S+s)^{-1}\overline{S^{-1/2}T^{1/2}}f_j\right\|s^{p/2}t^{p/2}\,\mathrm{d}s\,\mathrm{d}t\\
		&\leq C^{(p)}\|S^{-1/2}T^{1/2}\|^2\sum_{j}|\lambda_j|^2
		=C^{(p)}\|S^{-1/2}T^{1/2}\|^2\|K\|_\mathrm{HS}^2.
	\end{align*}
	Similarly, we obtain
	\begin{align*}
		&\int_0^\infty\int_0^\infty\sum_{j=1}^\infty\sum_{k=1}^\infty|\lambda_k|^2|\langle T^{(1-p)/2}(T+t)^{-1}f_k, T^{(1-p)/2}(T+s)^{-1}f_j\rangle|\\ 
		&\qquad\times |\langle S^{1/2}(S+s)^{-1}\overline{S^{-1/2}T^{1/2}}f_j,S^{1/2}(S+t)^{-1}\overline{S^{-1/2}T^{1/2}}f_k\rangle| s^{p/2}t^{p/2}\,\mathrm{d}s\,\mathrm{d}t\\
		&\leq C^{(p)}\|S^{-1/2}T^{1/2}\|^2\|K\|_\mathrm{HS}^2.
	\end{align*}
	Therefore we arrive at
	\[
	\sum_{n=1}^\infty\|Le_n\|^2\leq\frac{\sin^2{(\pi p/2)}}{\pi^2}C^{(p)}\|S^{-1/2}T^{1/2}\|^2\|K\|_\mathrm{HS}^2.
	\]
	This completes the proof.
\end{proof}

\begin{lemma}\label{conti of K wrt G}
	Let $S,T\in\mathcal{S\!A}_+(\mathscr{H})$ satisfy $S\overset{\mathrm{w}}{\sim} T,$ and let $p\in(0,1).$
	Let $\{S_n\}_{n=1}^\infty$ be a sequence in $\mathcal{S\!A}_+(\mathscr{H}).$
	Suppose that $S_n\overset{\mathrm{w}}{\sim} T$ for all $n\in\mathbb{N},$ 
	and that $\{K_{S_n,T}\}_{n=1}^\infty$ converges to $K_{S,T}$ in the operator norm.
	Then $\{L_{S_n^p,T^p}\}_{n=1}^\infty$ converges to $L_{S^p,T^p}$ in the operator norm as well.
    Moreover, we have
    \begin{equation}\label{sup of S_n^{-1/2}T^{1/2}}
    \sup_{n\in\mathbb{N}}\|S_n^{-1/2}T^{1/2}\|<\infty.
    \end{equation}
\end{lemma}

\begin{proof}
By direct computation, it holds that
	\begin{equation}\label{conti of K wrt G lim S_nS0}
	K_{S_n,S}=\overline{S^{-1/2}T^{1/2}}K_{S_n,T}\overline{T^{1/2}S^{-1/2}}+K_{T,S}
	\xrightarrow{n\to\infty}\overline{S^{-1/2}T^{1/2}}K_{S,T}\overline{T^{1/2}S^{-1/2}}+K_{T,S}=0,
	\end{equation}
	where the convergence is in the operator norm.
	Thus we have
	\[
	\overline{S^{1/2}S_n^{-1/2}}\cdot\left(\overline{S^{1/2}S_n^{-1/2}}\right)^*
	=\overline{S^{1/2}S_n^{-1/2}}\cdot\overline{S_n^{-1/2}S^{1/2}}
	=\left(K_{S_n,S}+1\right)^{-1}
	\xrightarrow{n\to\infty} 1
	\]
	in the operator norm.
	In particular, we obtain
	 \begin{equation}\label{conti of K wrt G supS_nS}
		\sup_{n\in\mathbb{N}}\|S_n^{-1/2}S^{1/2}\|=\sup_{n\in\mathbb{N}}\|(S_n^{-1/2}S^{1/2})^*\|=\sup_{n\in\mathbb{N}}\|S^{1/2}S_n^{-1/2}\|<\infty.
	\end{equation}
	The inequality \eqref{sup of S_n^{-1/2}T^{1/2}} follows from
	\[
	\sup_{n\in\mathbb{N}}\|S_n^{-1/2}T^{1/2}\|\leq\sup_{n\in\mathbb{N}}\|S_n^{-1/2}S^{1/2}\|\|S^{-1/2}T^{1/2}\|<\infty.
	\]
	For any $\psi,\phi\in\mathrm{dom}(T^{p/2})\cap\mathrm{dom}(T^{-p/2})$, Lemma \ref{integral rep power p ver2} yields that
	\begin{equation}\label{conti of K wrt G integral rep K1}
		\langle\psi,L_{S_n^p,T^p}\phi\rangle-\langle\psi,L_{S^p,T^p}\phi\rangle
		=\langle \psi,(S_n^{p/2}-S^{p/2})T^{-p/2}\phi\rangle
		=\frac{\sin{(\pi p/2)}}{\pi}\int_0^\infty \frac{f_n(t)}{t^{1-p/2}}\,\mathrm{d}t,
	\end{equation}
	where
	\[
	f_n(t):=\langle \psi,S_n(S_n+t)^{-1}T^{-p/2}\phi\rangle-\langle \psi,S(S+t)^{-1}T^{-p/2}\phi\rangle,\qquad t>0.
	\]
	Lemma \ref{diff of two resolvent} tells us that the function $f_n$ can be computed as 
	\begin{align*}
		f_n(t)&=-t\langle \psi,(S_n+t)^{-1}T^{-p/2}\phi\rangle+t\langle \psi,(S+t)^{-1}T^{-p/2}\phi\rangle\\
		&=t\langle \psi,S_n^{1/2}(S_n+t)^{-1}\overline{S_n^{-1/2}S^{1/2}}K_{S_n,S}S^{1/2}(S+t)^{-1}T^{-p/2}\phi\rangle,
	\end{align*}
and thus we obtain
	\begin{equation*}
	|f_n(t)|\leq t\|S_n^{1/2}(S_n+t)^{-1}\psi\|\|S_n^{-1/2}S^{1/2}\|\|K_{S_n,S}\|\|S^{(1-p)/2}(S+t)^{-1}\cdot S^{p/2}T^{-p/2}\phi\|.
\end{equation*}
By Lemma \ref{gen int formula1}, we have
\begin{equation*}
	\int_0^\infty\frac{|f_n(t)|}{t^{1-p/2}}\,\mathrm{d}t
	\leq (C^{(p)})^{1/2}\|S_n^{-1/2}S^{1/2}\|\|K_{S_n,S}\|\|\psi\|\|S^{p/2}T^{-p/2}\phi\|.
\end{equation*}
Combining this with \eqref{conti of K wrt G lim S_nS0}, \eqref{conti of K wrt G supS_nS}, and \eqref{conti of K wrt G integral rep K1}, we obtain
\[
\|L_{S_n^p,T^p}-L_{S^p,T^p}\|\leq (C^{(p)})^{1/2}\|S_n^{-1/2}S^{1/2}\|\|K_{S_n,S}\|\|S^{p/2}T^{-p/2}\|\xrightarrow{n\to\infty}0.
\]
This finishes the proof.
\end{proof}

\begin{proof}[Proof of Theorem \ref{main thm ver2}]
	Let $T=\int_0^\infty\lambda\,\mathrm{d}E_T(\lambda)$ be the spectral resolution of $T.$
	We write $K_{S,T}$ as
	\[
	K_{S,T}=\sum_n\lambda_n|f_n\rangle\langle f_n|,
	\]
	where $\{\lambda_n\}_{n=1}^\infty$ is a square-summable sequence of real numbers and $\{f_n\}_{n=1}^\infty$ is an orthonormal basis of $\mathscr{H}.$
	For each $N\in\mathbb{N},$ we denote by $P_N$ the orthogonal projection onto the subspace spanned by $f_1,\dots, f_N.$
	Let
	\begin{align*}
		U_{N}&:=\sum_{n=1}^N\lambda_n|T^{1/2}E_T([1/N,N])f_n\rangle\langle T^{1/2}E_T([1/N,N])f_n|\\
		&\ =T^{1/2}E_T([1/N,N])P_NK_{S,T}P_NT^{1/2}E_T([1/N,N]),
	\end{align*}
	and let
	\[
	S_N:=T+U_N.
	\]
	Since the operator
	\[
	K_{S,T}+1=\overline{T^{-1/2}S^{1/2}}\cdot\overline{S^{1/2}T^{-1/2}}=\left(\overline{S^{1/2}T^{-1/2}}\right)^*\overline{S^{1/2}T^{-1/2}}
	\]
	is non-negative and has a bounded inverse, 
	we can take a negative real constant $\gamma<0$ such that
	\[
	\inf\sigma(K_{S,T})>\gamma>-1.
	\]
	Then, for any $\psi\in\mathrm{dom}(T)$, it follows that
	\begin{align*}
		\langle\psi,S_N\psi\rangle
		&=\|T^{1/2}\psi\|^2+\langle P_NE_T([1/N,N])T^{1/2}\psi,K_{S,T}P_NE_T([1/N,N])T^{1/2}\psi\rangle\\
		&\geq\|T^{1/2}\psi\|^2+\gamma\|P_NE_T([1/N,N])T^{1/2}\psi\|^2
		\geq\|T^{1/2}\psi\|^2+\gamma\|T^{1/2}\psi\|^2\\
		&=(1+\gamma)\|T^{1/2}\psi\|^2,
	\end{align*}
	whence $S_N\in\mathcal{S\!A}_+(\mathscr{H})$.
	Moreover it holds that
	\begin{equation}\label{appx S_N equivalence1}
		(1+\gamma)T\preceq S_N.
	\end{equation}
	On the other hand, for any $\psi\in\mathrm{dom}(T)$, we have
	\begin{align*}
		\langle\psi,S_N\psi\rangle
		&=\|T^{1/2}\psi\|^2+\langle P_NE_T([1/N,N])T^{1/2}\psi,K_{S,T}P_NE_T([1/N,N])T^{1/2}\psi\rangle\\
		&\leq\|T^{1/2}\psi\|^2+\|K_{S,T}\|\|P_NE_T([1/N,N])T^{1/2}\psi\|^2\\
		&=(1+\|K_{S,T}\|)\|T^{1/2}\psi\|^2.
	\end{align*}
	This, together with the inequality \eqref{appx S_N equivalence1}, implies that
	\[
	(1+\gamma)T\preceq S_N\preceq(1+\|K_{S,T}\|)T,
	\]
	whence $S_N\overset{\mathrm{w}}{\sim} T.$
	
	We next compute $K_{S_N,T}.$ 
	Note that we have $\mathrm{dom}(S_N)=\mathrm{dom}(T)$,
	so for any $\psi,\phi\in\mathrm{dom}(T^{1/2})\cap\mathrm{dom}(T^{-1/2}),$ we have 
	\begin{align*}
		\langle\psi,K_{S_N,T}\phi\rangle
		&=\langle S_N^{1/2}T^{-1/2}\psi,S_N^{1/2}T^{-1/2}\phi\rangle-\langle\psi,\phi\rangle\\
		&=\langle T^{-1/2}\psi,S_NT^{-1/2}\phi\rangle-\langle\psi,\phi\rangle\\
		&=\langle \psi,E_T([1/N,N])P_NK_{S,T}P_NE_T([1/N,N])\phi\rangle,
	\end{align*}
	and thus we get
	\begin{equation}\label{main thm G_{T,S_N} formula}
		K_{S_N,T}=E_T([1/N,N])P_NK_{S,T}P_NE_T([1/N,N]).
	\end{equation}
	This, in particular, means that 
	\begin{equation}\label{main thm appx S_N HS norm ineq}
		\|K_{S_N,T}\|_{\mathrm{HS}}\leq\|K_{S,T}\|_{\mathrm{HS}},\qquad \forall N\in\mathbb{N}.
	\end{equation}
	Since both $\{E_T([1/N,N])\}_{N=1}^\infty$ and $\{P_N\}_{N=1}^\infty$ converge to the identity operator in the strong operator topology,
	and since $K_{S,T}$ is a compact operator, 
	it follows from the equality \eqref{main thm G_{T,S_N} formula} that $\{K_{S_N,T}\}_{N=1}^\infty$ converges to $K_{S,T}$ in the operator norm.
	This, together with Lemma \ref{conti of K wrt G}, yields that $\{L_{S_N^p,T^p}\}_{N=1}^\infty$ converges to $L_{S^p,T^p}$  in the operator norm.
	Moreover, we have
	\begin{equation}\label{uniform bddness ofS_N^{-1/2}T^{1/2}}
	M:=\sup_{N\in\mathbb{N}}\|S_N^{-1/2}T^{1/2}\|<\infty.
	\end{equation}
	
	We now use Lemma \ref{HS norm is lower semi-continuous}, Lemma \ref{main thm trace class ver}, the inequality \eqref{uniform bddness ofS_N^{-1/2}T^{1/2}} and the inequality \eqref{main thm appx S_N HS norm ineq} to get
	\begin{align*}
		\|L_{S^p,T^p}\|_\mathrm{HS}
		&\leq\liminf_{N\to\infty}\|L_{S_N^p,T^p}\|_\mathrm{HS}
		\leq\liminf_{N\to\infty}\frac{\sin{(\pi p/2)}}{\pi}(C^{(p)})^{1/2}\|S_N^{-1/2}T^{1/2}\|\|K_{S_N,T}\|_\mathrm{HS}\\
		&\leq\frac{\sin{(\pi p/2)}}{\pi}(C^{(p)})^{1/2}M\|K_{S,T}\|_\mathrm{HS}<\infty.
	\end{align*}
	Therefore $L_{S^p,T^p}$ is Hilbert--Schmidt.
\end{proof}

The main result of this section is the following, which is analogous to Proposition \ref{weak eq rel p}.

\begin{theorem}\label{main thm1}
	Let $S,T\in\mathcal{S\!A}_+(\mathscr{H})$ satisfy $S\sim T,$ and let $p\in(0,1].$
	Then $S^p\sim T^p$ and $S^{-p}\sim T^{-p}$.
\end{theorem}

\begin{lemma}\label{eq rel implies the inverse}
		Let $S,T\in\mathcal{S\!A}_+(\mathscr{H})$ satisfy $S\sim T.$
		Then $S^{-1}\sim T^{-1}$ holds with 
		\[
		K_{S^{-1},T^{-1}}=-K_{S,T}(K_{S,T}+1)^{-1}.
		\]
\end{lemma}

\begin{proof}
Since
\[
(K_{S,T}+1)^{-1}=\overline{T^{1/2}S^{-1/2}}\cdot\overline{S^{-1/2}T^{1/2}},
\]
we have
\[
K_{S^{-1},T^{-1}}=\overline{T^{1/2}S^{-1/2}}\cdot\overline{S^{-1/2}T^{1/2}}-1=(K_{S,T}+1)^{-1}-1=-K_{S,T}(K_{S,T}+1)^{-1}.
\]
The right-hand side is Hilbert--Schmidt, and thus $S^{-1}\sim T^{-1}$. 
\end{proof}

\begin{proof}[Proof of Theorem \ref{main thm1}]
	The case $p=1$ follows from Lemma \ref{eq rel implies the inverse}.
	In the following, we consider the case $p\in(0,1).$ 
	Theorem \ref{main thm ver2} and Proposition \ref{def of L} yield that $S^p\sim T^p$. 
	Combining this with Lemma \ref{eq rel implies the inverse}, we obtain $S^{-p}\sim T^{-p}$. 
\end{proof}

The following corollary generalizes \cite[Lemma 3.4]{MR4213757}.

\begin{corollary}\label{gener of MSU}
	Let $S,T\in\mathcal{S\!A}_+(\mathscr{H})$ satisfy $S^2\sim T^2.$
	Define a bounded operator $Y$ by
	\[
	Y:=\frac{1}{2}\left(\overline{T^{-1/2}S^{1/2}}-\overline{T^{1/2}S^{-1/2}}\right).
	\] 
	Then, $Y$ is Hilbert--Schmidt.
\end{corollary}

\begin{proof}
	Note that, by Theorem \ref{main thm1}, we have $S\sim T.$
	This in particular implies that $S\overset{\mathrm{w}}{\sim} T,$ and thus $Y$ is well-defined.
	We observe that
	\[
	4YY^*
	=\overline{T^{-1/2}S^{1/2}}\cdot\overline{S^{1/2}T^{-1/2}}+\overline{T^{1/2}S^{-1/2}}\cdot\overline{S^{-1/2}T^{1/2}}-2
	=K_{S,T}+K_{S^{-1},T^{-1}}.
	\]
	By Lemma \ref{eq rel implies the inverse}, we have
	\[
	4YY^*
	=K_{S,T}-K_{S,T}(K_{S,T}+1)^{-1}
	=(K_{S,T})^2(K_{S,T}+1)^{-1}.
	\]
	Since $K_{S,T}$ is Hilbert--Schmidt, $YY^*$ is of trace class.
	Hence $Y$ is Hilbert--Schmidt.
\end{proof}

The following proposition gives a necessary condition for two self-adjoint operators to be equivalent.

\begin{proposition}\label{eq necessary condition HS}
	Let $S,T\in\mathcal{S\!A}_+(\mathscr{H})$ satisfy $S\sim T,$ and let $z\in\mathbb{C}$ be such that $z\not\in\sigma(S)\cup\sigma(T)$.
	Then, the operator
	\[
	(S-z)^{-1}-(T-z)^{-1}
	\]
	is Hilbert--Schmidt.
	In particular, we have $\sigma_\mathrm{ess}(S)=\sigma_\mathrm{ess}(T).$
\end{proposition}

\begin{proof}
This follows from Lemma \ref{diff of two resolvent} and Lemma \ref{lem;ess_spec}.
\end{proof}

\begin{corollary}\label{S^{-2}-T^{-2} is HS}
	Let $S,T\in\mathcal{S\!A}_+(\mathscr{H})$ satisfy $S\sim T.$
	Suppose that $S$ and $T$ have bounded inverses.
	Then $S^{-2}-T^{-2}$ is Hilbert--Schmidt.
	In particular, we have $\sigma_\mathrm{ess}(S^2)=\sigma_\mathrm{ess}(T^2).$
\end{corollary}

\begin{proof}
	We have
	\[
	S^{-2}-T^{-2}=S^{-1}(S^{-1}-T^{-1})+(S^{-1}-T^{-1})T^{-1}.
	\]
	By Proposition \ref{eq necessary condition HS}, the right-hand side is Hilbert--Schmidt.
	Moreover, since $0\not\in\sigma(S^2)\cup\sigma(T^2)$, we can write
	\[
	S^{-2}-T^{-2}=(S^2-0)^{-1}-(T^2-0)^{-1}.
	\]
    Hence, by Lemma \ref{lem;ess_spec}, we obtain $\sigma_\mathrm{ess}(S^2)=\sigma_\mathrm{ess}(T^2)$.
\end{proof}

\subsection{A related question}
Let $S,T\in\mathcal{S\!A}_+(\mathscr{H})$ satisfy $S^2\overset{\mathrm{w}}{\sim} T^2.$
Then, it follows from Theorem \ref{main thm1} that $S^2\sim T^2$ implies $S\sim T.$
The converse may also be true:

\begin{question}
	Let $S,T\in\mathcal{S\!A}_+(\mathscr{H})$ satisfy $S^2\overset{\mathrm{w}}{\sim} T^2$ and $S\sim T.$
	Is it true that $S^2\sim T^2$? 
\end{question}

The reason we expect this to be true is that it holds in the following two special cases.

\begin{proposition}
		Let $S,T\in\mathcal{S\!A}_+(\mathscr{H})$ satisfy $S^2\overset{\mathrm{w}}{\sim} T^2$ and $S\sim T.$
		If $S$ and $T$ strongly commute, meaning that their spectral projections $E_{S}(B)$ and $E_{T}(B)$ commute for all Borel sets $B$ of $\mathbb{R},$ then $S^2\sim T^2$ holds.
\end{proposition}

\begin{proof}
	Since $S$ and $T$ strongly commute, we have
	\[
	K_{S,T}=\overline{T^{-1/2}S^{1/2}}\cdot\overline{S^{1/2}T^{-1/2}}-1=\overline{ST^{-1}}-1=L_{S^2,T^2}.
	\]
	Thus, $L_{S^2,T^2}$ is Hilbert--Schmidt.
	By Proposition \ref{def of L}, we conclude that $S^2\sim T^2.$
\end{proof}

\begin{proposition}
	Let $S,T\in\mathcal{S\!A}_+(\mathscr{H})$ satisfy $S^2\overset{\mathrm{w}}{\sim} T^2$ and $S\sim T.$
	If $T$ is bounded and has a bounded inverse, then $S^2\sim T^2$ holds.
\end{proposition}

\begin{proof}
	Since both $T^{1/2}$ and $T^{-1/2}$ are bounded, we observe that
	\[
	T^{1/2}K_{S,T}T^{-1/2}=T^{1/2}\overline{T^{-1/2}S^{1/2}}\cdot\overline{S^{1/2}T^{-1/2}}T^{-1/2}-1=ST^{-1}-1=L_{S^2,T^2}
	\]
	is Hilbert--Schmidt.
	Thus, by Proposition \ref{def of L}, we conclude that $S^2\sim T^2$.
\end{proof}

\section{Subspaces spread by self-adjoint operators}\label{Subspaces spread by self-adjoint operators}
In this section, we introduce a tool for constructing Weyl representations.
Let $\mathscr{H}$ be an infinite-dimensional, separable complex Hilbert space.

\begin{definition}\label{def of spread subsp}
	Let $T\in\mathcal{S\!A}_+(\mathscr{H}),$ and let $\mathscr{D}$ be a complex subspace of $\mathscr{H}.$
	We say that $\mathscr{D}$ is \textit{spread} by $T$ if the following two conditions hold:
	\begin{itemize}
	\item $\mathscr{D}\subset \mathrm{dom}(T^{1/2})\cap\mathrm{dom}(T^{-1/2}),$
	\item $T^{1/2}\mathscr{D}$ and $T^{-1/2}\mathscr{D}$ are dense in $\mathscr{H}.$
	\end{itemize}
\end{definition}

\begin{example}
	Let $T\in\mathcal{S\!A}_+(\mathscr{H}).$
	Then $\mathscr{D}:= \mathrm{dom}(T^{1/2})\cap\mathrm{dom}(T^{-1/2})$ is spread by $T.$
	Indeed, we have
	\[
	T^{\pm1/2}\mathscr{D}\supset T^{\pm1/2}\cdot T^{\mp1/2}\mathrm{dom}(T^{\mp1})=\mathrm{dom}(T^{\mp1}),
	\]
	and thus $T^{\pm1/2}\mathscr{D}$ is dense in $\mathscr{H}.$
\end{example}

The spread property is preserved under the weak equivalence relation.

\begin{proposition}\label{spread preserve we}
	Let $S,T\in\mathcal{S\!A}_+(\mathscr{H})$ satisfy $S\overset{\mathrm{w}}{\sim} T,$ and let  $\mathscr{D}$ be a complex subspace of $\mathscr{H}.$
	If $\mathscr{D}$ is spread by $T,$ then it is also spread by $S.$
\end{proposition}

\begin{proof}
	By Proposition \ref{weq iff}, we have
	\[
	\mathscr{D}\subset \mathrm{dom}(T^{1/2})\cap\mathrm{dom}(T^{-1/2})=\mathrm{dom}(S^{1/2})\cap\mathrm{dom}(S^{-1/2}),
	\]
	and thus the first condition in Definition \ref{def of spread subsp} holds.
	
	We next show the second condition.
	It follows from Proposition \ref{weq iff} that $\overline{S^{\pm 1/2}T^{\mp1/2}}$ is bounded and has a bounded inverse, whence
	\[
	S^{\pm 1/2}\mathscr{D}=\overline{S^{\pm 1/2}T^{\mp1/2}}\cdot T^{\pm1/2}\mathscr{D}
	\]
	is dense in $\mathscr{H}.$
	This finishes the proof.
\end{proof}

Recall that the symbol $\mathcal{S\!A}_{C,\mathscr{V}}(\mathscr{H})$ for a real subspace $\mathscr{V}$ of $\mathscr{H}_C$ is defined in Subsection \ref{subsec weyl}.

\begin{proposition}\label{D_J is dense}
	Let $C$ be a conjugation on $\mathscr{H}.$
	Let $T\in\mathcal{S\!A}_+(\mathscr{H})$ satisfy $CT\subset TC,$
	and let $\mathscr{D}$ be a complex subspace of $\mathscr{H}$ satisfying $C\mathscr{D}\subset\mathscr{D}.$
	Suppose that $\mathscr{D}$ is spread by $T.$
	Then $T^{1/2}$ belongs to $\mathcal{S\!A}_{C,\mathscr{D}_C}(\mathscr{H}).$ 
\end{proposition}

\begin{proof}
	It suffices to show that $T^{1/2}\mathscr{D}_C$ and $T^{-1/2}\mathscr{D}_C$ are dense real subspaces of $\mathscr{H}_C.$
	For any $f\in\mathscr{D}_C$, we have
	\[
	C\cdot T^{\pm1/2}f=T^{\pm1/2}Cf=T^{\pm1/2}f,
	\]
	and thus $T^{\pm1/2}f$ is in $\mathscr{H}_C.$
	In particular, $T^{1/2}\mathscr{D}_C$ and $T^{-1/2}\mathscr{D}_C$ are subsets of $\mathscr{H}_C.$
	It is easy to check that $T^{1/2}\mathscr{D}_C$ and $T^{-1/2}\mathscr{D}_C$ are real subspaces of $\mathscr{H}_C.$
	
	The remaining task is to show that $T^{1/2}\mathscr{D}_C$ and $T^{-1/2}\mathscr{D}_C$ are dense in $\mathscr{H}_C.$
	Let $f\in\mathscr{H}_C$ be arbitrary.
	Since $T^{\pm1/2}\mathscr{D}$ is dense in $\mathscr{H},$ there exists a sequence $\{f_n\}_{n=1}^\infty$ in $\mathscr{D}$ such that $\{T^{\pm1/2}f_n\}_{n=1}^\infty$ converges to $f.$
	Define $g_n:=(f_n+Cf_n)/2.$
	By assumption, we have $g_n\in\mathscr{D}_C.$ 
	Moreover, it follows that
	\[
	T^{\pm1/2}g_n=\frac{1}{2}(T^{\pm1/2}f_n+CT^{\pm1/2}f_n)\xrightarrow{n\to\infty}\frac{1}{2}(f+Cf)=f.
	\]
	This completes the proof.
\end{proof}

\section{Quantization of Klein-Gordon equation with an external potential}\label{Quantization of Klein-Gordon equation with an external potential}

In this section, we first explicitly construct our model and then prove the main results of this paper.

Let $\mathscr{H}:=L^2(\mathbb{R}^3,\mathrm{d}x).$ 
The symbol $\mathcal{S}(\mathbb{R}^3)$ denotes the Schwartz space on $\mathbb{R}^3,$
and the symbol $\mathcal{S}(\mathbb{R}^3;\mathbb{R})$ denotes the set of real-valued Schwartz functions on $\mathbb{R}^3.$
Let $C$ be a conjugation on $\mathscr{H}$ defined by
\[
(Cf)(x):=\overline{f(x)},\qquad f\in\mathscr{H},\ \textrm{a.e.}\,x\in\mathbb{R}^3.
\]
It is easy to check that $C\mathcal{S}(\mathbb{R}^3)\subset\mathcal{S}(\mathbb{R}^3)$ and $\mathcal{S}(\mathbb{R}^3)_C=\mathcal{S}(\mathbb{R}^3;\mathbb{R}).$

Let $m>0$ be a positive real number.
We define a self-adjoint operator $T_{m,0}$ by $T_{m,0}:=\sqrt{-\Delta+m^2},$ where $\Delta$ is the generalized Laplacian that acts in $\mathscr{H}.$
Then $T_{m,0}$ belongs to $\mathcal{S\!A}_+(\mathscr{H}).$
By making use of the Fourier transform, one can show that $\mathcal{S}(\mathbb{R}^3)$ is spread by $T_{m,0}.$

Let $\mathcal{R}$ be the Rollnik class:
\[
\mathcal{R}:=\left\{V:\mathbb{R}^3\to\mathbb{R}\ \text{Borel}\ \middle|\  \int_{\mathbb{R}^3}\int_{\mathbb{R}^3}\frac{|V(x)||V(y)|}{|x-y|^2}\,\mathrm{d}x\,\mathrm{d}y<\infty\right\}.
\]
Note that $\mathcal{R}$ contains $L^{3/2}(\mathbb{R}^3)$ \cite[Theorem I.1]{MR455975}, and that $\mathcal{R}$ is a real vector space \cite[Theorem I.17]{MR455975}.
The symbol $\mathcal{R}+L_\varepsilon^\infty$ denotes the set of all Borel functions $V:\mathbb{R}^3\to\mathbb{R}$ such that
for every $\varepsilon>0,$ there exist $V_1\in\mathcal{R}$ and $ V_2\in L^\infty(\mathbb{R}^3)$ satisfying $V=V_1+V_2$ and $\|V_2\|_{L^\infty(\mathbb{R}^3)}<\varepsilon.$
Similarly, the symbol $L^2+L_\varepsilon^\infty$ denotes the set of real-valued Borel functions $V$ that satisfy the condition in the definition of $\mathcal{R}+L_\varepsilon^\infty,$ with $\mathcal{R}$ replaced by $L^2(\mathbb{R}^3).$
It is known that $L^2+L_\varepsilon^\infty$ is a subset of $\mathcal{R}+L_\varepsilon^\infty.$
The proof can be found in \cite[Remark below Corollary I.2]{MR455975}.

Let $V\in\mathcal{R}+L_\varepsilon^\infty$ be arbitrary.
Then, $V$ is infinitesimally form bounded with respect to $-\Delta.$ 
Hence, by the KLMN theorem, the self-adjoint operator $-\Delta+m^2+V$ is well-defined as a form sum \cite[Corollary II.8]{MR455975}.
Note that its form domain coincides with $\mathrm{dom}(\sqrt{-\Delta}).$ 
Moreover, for any sufficiently large $E>0,$ the difference
\[
(-\Delta+m^2+V+E)^{-1}-(-\Delta+m^2+E)^{-1}
\]
is a compact operator \cite[Example 7 in Section  XIII.4]{MR493421}, from which it follows that
\begin{equation}\label{ess spec of T_{m,V}^2}
	\sigma_\mathrm{ess}(-\Delta+m^2+V)=[m^2,\infty).
\end{equation}

We now assume that $-\Delta+m^2+V$ belongs to $\mathcal{S\!A}_+(\mathscr{H}),$ and let $T_{m,V}$ denote its square root:
\[
T_{m,V}:=\sqrt{-\Delta+m^2+V}.
\]
By \eqref{ess spec of T_{m,V}^2} and the assumption that $(T_{m,V})^2$ is injective, $(T_{m,V})^2$ has a bounded inverse.
This, together with Proposition \ref{we suf cond for massive ops} implies that $(T_{m,V})^2\overset{\mathrm{w}}{\sim} (T_{m,0})^2.$
By Proposition \ref{spread preserve we}, $\mathcal{S}(\mathbb{R}^3)$ is spread by $T_{m,V}.$ 
Since $CT_{m,V}\subset T_{m,V}C$, Proposition \ref{D_J is dense} yields that $T_{m,V}^{1/2}$ belongs to $\mathcal{S\!A}_{C,\mathcal{S}(\mathbb{R}^3;\mathbb{R})}(\mathscr{H}).$
Thus, the triplet
\[
\left\{\mathscr{F}_{\mathrm{b}}(\mathscr{H}),\mathrm{d}\Gamma_\mathrm{b}(T_{m,V}),\{\phi_{T_{m,V}^{1/2}}(f),\pi_{T_{m,V}^{1/2}}(f)\mid f\in\mathcal{S}(\mathbb{R}^3;\mathbb{R})\}\right\}
\]
is an abstract free Bose field model.
It follows from Proposition \ref{AFBFM eq} that the model satisfies the following differential equations:
\[
\frac{\mathrm{d}}{\mathrm{d}t}\phi_{T_{m,V}^{1/2}}(t,f)\Psi=\pi_{T_{m,V}^{1/2}}(t,f)\Psi,\qquad \frac{\mathrm{d}^2}{\mathrm{d}t^2}\phi_{T_{m,V}^{1/2}}(t,f)\Psi+\phi_{T_{m,V}^{1/2}}(t,(-\Delta+m^2+V)f)\Psi=0
\]
for all $\Psi\in\mathrm{dom}(\mathrm{d}\Gamma_{\mathrm{b}}(1)^{1/2})$ and $f\in\mathcal{S}(\mathbb{R}^3;\mathbb{R}).$
Hence, we have constructed a quantized scalar field model obeying the field equation \eqref{KGeq}.

The aim of this section is to study the Weyl representation
\[
\rho_{m,V}:=\rho_{T_{m,V}^{1/2}}=\left\{\mathscr{F}_{\mathrm{b}}(\mathscr{H}),\{\phi_{T_{m,V}^{1/2}}(f),\pi_{T_{m,V}^{1/2}}(f)\mid f\in\mathcal{S}(\mathbb{R}^3;\mathbb{R})\}\right\}
\]
corresponding to the model.
Our first result generalizes the well-known theorem that $\rho_{m_1,0}$ is not equivalent to $\rho_{m_2,0}$ whenever $m_1\not=m_2$ (see e.g., \cite[Theorem 1.1]{MR3513945}, \cite[Theorem 10.13]{MR4292535} or \cite[Theorem X.46]{MR493420}).
It is stated as follows:

\begin{theorem}\label{m_1not=m_2}
	Let $m_1,m_2>0,$ and let $V_1,V_2\in\mathcal{R}+L_\varepsilon^\infty.$ 
	Suppose that $-\Delta+m_j^2+V_j$ belongs to $\mathcal{S\!A}_+(\mathscr{H})$ for each $j=1,2.$
	Then the following statements hold true:
	\begin{itemize}
		\item[\textup{(1)}] $(T_{m_1,V_1})^2\overset{\mathrm{w}}{\sim} (T_{m_2,V_2})^2,$
		\item[\textup{(2)}] a transfer pair $(J_+,J_-)$ from $T_{m_1,V_1}^{1/2}$ to $T_{m_2,V_2}^{1/2}$ with respect to $(C,\mathcal{S}(\mathbb{R}^3;\mathbb{R}))$ exists, and is given by
		\begin{equation}\label{J_+=T_{m_2,V_2}^{1/2}T_{m_1,V_1}^{-1/2}}
			J_+=T_{m_2,V_2}^{1/2}T_{m_1,V_1}^{-1/2},\qquad J_-=\overline{T_{m_2,V_2}^{-1/2}T_{m_1,V_1}^{1/2}},
		\end{equation}
		\item[\textup{(3)}] $\rho_{m_1,V_1}$ is equivalent to $\rho_{m_2,V_2}$ if and only if $T_{m_1,V_1}\sim T_{m_2,V_2},$
		\item[\textup{(4)}] if $m_1\not=m_2,$ then $\rho_{m_1,V_1}$ is not equivalent to $\rho_{m_2,V_2}.$
	\end{itemize}
\end{theorem}

\begin{proof}
	We first prove (1).
	By Proposition \ref{we suf cond for massive ops}, we have $(T_{{m_1},0})^2\overset{\mathrm{w}}{\sim} (T_{m_2,0})^2,$ and thus
	\[
	(T_{{m_1},V_1})^2 \overset{\mathrm{w}}{\sim} (T_{{m_1},0})^2 \overset{\mathrm{w}}{\sim} (T_{{m_2},0})^2 \overset{\mathrm{w}}{\sim} (T_{{m_2},V_2})^2.
	\]
	
	The operators $J_+$ and $J_-$ defined in \eqref{J_+=T_{m_2,V_2}^{1/2}T_{m_1,V_1}^{-1/2}} satisfy the conditions in the definition of a transfer pair, whence (2) follows.
	
	We next prove (3).
	Since
	\[
	K_{T_{{m_1},V_1},T_{{m_2},V_2}}=(J_--J_+)T_{m_1,V_1}^{1/2}T_{m_2,V_2}^{-1/2},
	\]
	and since $T_{m_1,V_1}^{1/2}T_{m_2,V_2}^{-1/2}$ has a bounded inverse, $T_{m_1,V_1}\sim T_{m_2,V_2}$ if and only if the difference $J_+-J_-$ is Hilbert--Schmidt.
	Thus (3) follows from Theorem \ref{eq or ineq weyl repr}.

	Finally, we prove the contrapositive of (4).
	Suppose that $\rho_{m_1,V_1}$ is equivalent to $\rho_{m_2,V_2}.$
	By (3), we have $T_{m_1,V_1}\sim T_{m_2,V_2}$.
	This, together with Corollary \ref{S^{-2}-T^{-2} is HS}, yields that the essential spectrum of $(T_{m_1,V_1})^2$ coincides with that of $(T_{m_2,V_2})^2.$
	By \eqref{ess spec of T_{m,V}^2}, we conclude that $m_1=m_2.$
	This completes the proof.
\end{proof}

\begin{remark}
	Arai \cite[Theorem 6.5]{MR3513945} studied a Weyl representation corresponding to the case where $m_1=0$ and $V_1=0,$ and proved that it is not equivalent to $\rho_{m_2,0}$ for any $m_2>0.$
\end{remark}

In the following, we consider the case where $m_1=m_2>0$.
Let $\mathcal{R}+L^2$ denote the set of all Borel functions $V:\mathbb{R}^3\to\mathbb{R}$ for which there exist $ V_1\in \mathcal{R}$ and $V_2\in L^2(\mathbb{R}^3)$ such that $V=V_1+V_2.$ 
Note that $\mathcal{R}+L^2\subset\mathcal{R}+L_\varepsilon^\infty$, since $L^2+L_\varepsilon^\infty\subset\mathcal{R}+L_\varepsilon^\infty$.

The following theorem generalizes the results of \cite[Lemma 2]{MR356771} and \cite[Lemma 2.2]{MR2840105}.

\begin{theorem}\label{T_{m,V} eq if R+L^2}
	Let $m>0,$ and let $V_1,V_2\in\mathcal{R}+L_\varepsilon^\infty.$ 
	Suppose that $-\Delta+m^2+V_j$ belongs to $\mathcal{S\!A}_+(\mathscr{H})$ for each $j=1,2.$
	If $V_1-V_2$ is in $\mathcal{R}+L^2,$ then $\rho_{m,V_1}$ is equivalent to $\rho_{m,V_2}.$ 
\end{theorem}

\begin{lemma}\label{lemma for T_{m,V} eq if R+L^2}
If $V\in\mathcal{R}+L^2,$ then 
\begin{equation}\label{int of lemma for T_{m,V} eq if R+L^2}
\int_{\mathbb{R}^3}\int_{\mathbb{R}^3}|V(x)|\frac{\mathrm{e}^{-2m|x-y|}}{|x-y|^2}|V(y)|\,\mathrm{d}x\,\mathrm{d}y<\infty.
\end{equation}
\end{lemma}

\begin{proof}
	Since the function $g_m:=\mathrm{e}^{-2m|\cdot|}/|\cdot|^2$ is integrable over $\mathbb{R}^3,$ Young's convolution inequality tells us that \eqref{int of lemma for T_{m,V} eq if R+L^2} holds for any $V\in L^2(\mathbb{R}^3).$
	On the other hand, by the definition of the Rollnik class, \eqref{int of lemma for T_{m,V} eq if R+L^2} holds for every $V\in \mathcal{R}.$
	Thus, the remaining task is to show that the set of all Borel functions $V$ satisfying \eqref{int of lemma for T_{m,V} eq if R+L^2} is closed under addition.
	
	Let $V$ be a non-negative Borel function in $L^1(\mathbb{R}^3)\cap L^2(\mathbb{R}^3).$
	Since $L^1(\mathbb{R}^3)\cap L^2(\mathbb{R}^3)$ is a subset of $\mathcal{R}$ \cite[Corollary I.3]{MR455975}, the inequality \eqref{int of lemma for T_{m,V} eq if R+L^2} holds.
	For any $f\in L^1(\mathbb{R}^3),$ we define its Fourier transform $\hat{f}$ by
	\[
	\hat{f}(k):=\frac{1}{(2\pi)^{3/2}}\int_{\mathbb{R}^3}f(x)\mathrm{e}^{-\mathrm{i}x\cdot k}\,\mathrm{d}x,\qquad k\in\mathbb{R}^3.
	\]
	The Fourier transform $\mathscr{F}$ on $L^2(\mathbb{R}^3)$ is defined in the same way, but in the $L^2$-sense.
	Then we obtain
	\begin{align*}
		&\int_{\mathbb{R}^3}\int_{\mathbb{R}^3}|V(x)|\frac{\mathrm{e}^{-2m|x-y|}}{|x-y|^2}|V(y)|\,\mathrm{d}x\,\mathrm{d}y
		=\langle V,g_m*V\rangle_{L^2(\mathbb{R}^3)}
		=\langle \mathscr{F}V,\mathscr{F}(g_m*V)\rangle_{L^2(\mathbb{R}^3)}\\
		&\qquad\qquad\qquad=(2\pi)^{3/2}\langle \mathscr{F}V,\widehat{g_m}(\mathscr{F}V)\rangle_{L^2(\mathbb{R}^3)}
		=(2\pi)^{3/2}\int_{\mathbb{R}^3}\widehat{g_m}(k)|(\mathscr{F}V)(k)|^2\,\mathrm{d}k.
	\end{align*}
Note that
\[
(2\pi)^{3/2}\widehat{g_m}(k)=\frac{4\pi}{|k|}\arctan{\frac{|k|}{2m}}\geq0,\qquad \forall k\in\mathbb{R}^3.
\]

We now consider arbitrary real-valued Borel functions $V_1,V_2$ satisfying \eqref{int of lemma for T_{m,V} eq if R+L^2}.
For each $j=1,2,$ we choose a sequence $\{f_{j,n}\}_{n=1}^\infty$ of integrable non-negative simple functions such that
\[
0\leq f_{j,n}(x)\leq f_{j,n+1}(x)\leq |V_j(x)|,\qquad \lim_{n\to\infty} f_{j,n}(x)=|V_{j}(x)|,\qquad \forall x\in\mathbb{R}^3.
\]
Since each $f_{j,n}$ belongs to $L^1(\mathbb{R}^3)\cap L^2(\mathbb{R}^3),$ the monotone convergence theorem yields that
\begin{align*}
	&\int_{\mathbb{R}^3}\int_{\mathbb{R}^3}|V_1(x)+V_2(x)|\frac{\mathrm{e}^{-2m|x-y|}}{|x-y|^2}|V_1(y)+V_2(y)|\,\mathrm{d}x\,\mathrm{d}y\\
	&\leq\int_{\mathbb{R}^3}\int_{\mathbb{R}^3}\big(|V_1(x)|+|V_2(x)|\big)\frac{\mathrm{e}^{-2m|x-y|}}{|x-y|^2}\big(|V_1(y)|+|V_2(y)|\big)\,\mathrm{d}x\,\mathrm{d}y\\
	&=\lim_{n\to\infty}\int_{\mathbb{R}^3}\int_{\mathbb{R}^3}\big(f_{1,n}(x)+f_{2,n}(x)\big)\frac{\mathrm{e}^{-2m|x-y|}}{|x-y|^2}\big(f_{1,n}(y)+f_{2,n}(y)\big)\,\mathrm{d}x\,\mathrm{d}y\\
	&=\lim_{n\to\infty}	(2\pi)^{3/2}\int_{\mathbb{R}^3}\widehat{g_m}(k)|(\mathscr{F}f_{1,n})(k)+(\mathscr{F}f_{2,n})(k)|^2\,\mathrm{d}k\\
	&\leq\lim_{n\to\infty}(2\pi)^{3/2}\left(2\int_{\mathbb{R}^3}\widehat{g_m}(k)|(\mathscr{F}f_{1,n})(k)|^2\,\mathrm{d}k+2\int_{\mathbb{R}^3}\widehat{g_m}(k)|(\mathscr{F}f_{2,n})(k)|^2\,\mathrm{d}k\right)\\
	&=2\int_{\mathbb{R}^3}\int_{\mathbb{R}^3}|V_1(x)|\frac{\mathrm{e}^{-2m|x-y|}}{|x-y|^2}|V_1(y)|\,\mathrm{d}x\,\mathrm{d}y
	+2\int_{\mathbb{R}^3}\int_{\mathbb{R}^3}|V_2(x)|\frac{\mathrm{e}^{-2m|x-y|}}{|x-y|^2}|V_2(y)|\,\mathrm{d}x\,\mathrm{d}y,
\end{align*}
where the right-hand side is finite.
Thus, $V_1+V_2$ also satisfies \eqref{int of lemma for T_{m,V} eq if R+L^2}.
This completes the proof.
\end{proof}

\begin{remark}
	Under the assumption $V\in\mathcal{R}+L_\varepsilon^\infty,$ the authors do not know whether the condition $V\in\mathcal{R}+L^2$ is necessary for \eqref{int of lemma for T_{m,V} eq if R+L^2} to hold.
\end{remark}

\begin{proof}[Proof of Theorem \ref{T_{m,V} eq if R+L^2}]
	We show that $(T_{m,V_1})^2\sim (T_{m,V_2})^2$ whenever $V_1-V_2$ belongs to $\mathcal{R}+L^2.$
	Once this is established, Theorem \ref{main thm1} implies that $T_{m,V_1}\sim T_{m,V_2},$ and thus, by Theorem \ref{m_1not=m_2}, $\rho_{m,V_1}$ is equivalent to $\rho_{m,V_2}.$
	
	Let $V:=V_1-V_2.$
	For each $j=1,2,$ it follows from the definition of $T_{m,V_j}$ that 
	\[
	\|T_{m,V_j}\psi\|^2=\|\sqrt{-\Delta}\psi\|^2+m^2\|\psi\|^2+\int_{\mathbb{R}^3}V_j(x)|\psi(x)|^2\,\mathrm{d}x,\qquad\forall\psi\in\mathrm{dom}(\sqrt{-\Delta}).
	\]
	Thus, for any $\phi\in\mathscr{H},$ we have
	\begin{align*}
		\langle \phi,K_{(T_{m,V_1})^2,(T_{m,V_2})^2}\phi\rangle
		&=\|T_{m,V_1}T_{m,V_2}^{-1}\phi\|^2- \|T_{m,V_2}T_{m,V_2}^{-1}\phi\|^2
		=\int_{\mathbb{R}^3}V(x)|(T_{m,V_2}^{-1}\phi)(x)|^2\,\mathrm{d}x\\
		&=\langle |V|^{1/2}T_{m,V_2}^{-1}\phi,\mathrm{sgn}(V)|V|^{1/2}T_{m,V_2}^{-1}\phi\rangle,
	\end{align*}
where $\mathrm{sgn}:\mathbb{R}\to\mathbb{R}$ denotes the sign function, defined by
\[
\mathrm{sgn}(x):=\begin{cases}
	1,\qquad & x>0,\\
	0, & x=0,\\
	-1, & x<0. 
\end{cases}
\]
Hence, we obtain
\begin{align}\label{formula of K_{(T_{m,V_1})^2,(T_{m,V_2})^2}}
K_{(T_{m,V_1})^2,(T_{m,V_2})^2}
&=\left(|V|^{1/2}T_{m,V_2}^{-1}\right)^*\mathrm{sgn}(V)\cdot|V|^{1/2}T_{m,V_2}^{-1} \notag\\
&=\left(T_{m,0}T_{m,V_2}^{-1}\right)^*\left(|V|^{1/2}T_{m,0}^{-1}\right)^*\mathrm{sgn}(V)\cdot|V|^{1/2}T_{m,0}^{-1}\cdot T_{m,0}T_{m,V_2}^{-1}.
\end{align}
To prove that $K_{(T_{m,V_1})^2,(T_{m,V_2})^2}$ is Hilbert--Schmidt, we first show that 
\begin{equation}\label{|V|^{1/2}T_{m,0}^{-2}|V|^{1/2} def}
|V|^{1/2}T_{m,0}^{-1}\left(|V|^{1/2}T_{m,0}^{-1}\right)^*=\overline{|V|^{1/2}T_{m,0}^{-2}|V|^{1/2}}
\end{equation}
is Hilbert--Schmidt.
Note that the domain of $|V|^{1/2}T_{m,0}^{-2}|V|^{1/2}$ coincides with that of $|V|^{1/2},$ and thus $|V|^{1/2}T_{m,0}^{-2}|V|^{1/2}$ is densely defined.
Since $\overline{|V|^{1/2}T_{m,0}^{-2}|V|^{1/2}}$ is an integral operator with kernel
\[
k_{m}(x,y):=\frac{1}{4\pi}|V(x)|^{1/2}\frac{\mathrm{e}^{-m|x-y|}}{|x-y|}|V(y)|^{1/2},\qquad x,y\in\mathbb{R}^3,
\]
it is Hilbert--Schmidt if and only if $k_m$ is square-integrable over $\mathbb{R}^3\times\mathbb{R}^3.$
By Lemma \ref{lemma for T_{m,V} eq if R+L^2}, the kernel $k_m$ is indeed square-integrable.
Thus, $\overline{|V|^{1/2}T_{m,0}^{-2}|V|^{1/2}}$ is Hilbert--Schmidt.

By \eqref{|V|^{1/2}T_{m,0}^{-2}|V|^{1/2} def}, we deduce that $\overline{|V|^{1/2}T_{m,0}^{-1}}$ belongs to the set $\mathcal{C}_4(\mathscr{H})$ of Schatten $4$-class operators.
This, together with \eqref{formula of K_{(T_{m,V_1})^2,(T_{m,V_2})^2}} and Lemma \ref{noncommutative Holder ineq}, implies that $K_{(T_{m,V_1})^2,(T_{m,V_2})^2}$ is Hilbert--Schmidt.
This completes the proof.
\end{proof}

In the following theorem, we consider potentials belonging to $L^2+L_\varepsilon^\infty$ rather than $\mathcal{R}+L_\varepsilon^\infty$ for a technical reason.

\begin{theorem}\label{T_{m,V} ineq}
	Let $m>0,$ and let $V_1,V_2\in L^2+L_\varepsilon^\infty.$ 
	Suppose that $-\Delta+m^2+V_j$ belongs to $\mathcal{S\!A}_+(\mathscr{H})$ for each $j=1,2.$
	If there exist a positive real number $R>0$ and a non-increasing Borel function $f:[R,\infty)\to(0,\infty)$ such that
	\[
	\int_R^\infty r^2f(r)^2\,\mathrm{d}r=\infty,\qquad V_1(x)-V_2(x)\geq f(|x|),\qquad\forall \,|x|\geq R,
	\]
	then $\rho_{m,V_1}$ is not equivalent to $\rho_{m,V_2}.$ 
\end{theorem}

\begin{lemma}\label{lemma for T_{m,V} ineq}
	For any $\psi\in\mathscr{H},$ we have
	\[
	\left[(-\Delta+m^2)^{-2}\psi\right](x)=\frac{1}{8\pi m}\int_{\mathbb{R}^3}\mathrm{e}^{-m|x-y|}\psi(y)\,\mathrm{d}y,\qquad\mathrm{a.e.}\,x\in\mathbb{R}^3.
	\]
\end{lemma}

\begin{proof}
   Define a function $g_m$ by
	\[
	g_m(x):=\frac{\mathrm{e}^{-m|x|}}{4\pi|x|},\qquad x\in\mathbb{R}^3.
	\] 
	Then $g_m$ is in $L^1(\mathbb{R}^3)\cap L^2(\mathbb{R}^3)$.
	Since $(-\Delta+m^2)^{-1}\phi=g_m*\phi$ for any $\phi\in\mathscr{H}$, it follows that
	\[
	(-\Delta+m^2)^{-2}\psi=g_m*(g_m*\psi)=(g_m*g_m)*\psi.
	\]
	Thus, it is sufficient to prove that
	\begin{equation}\label{lemma for T_{m,V} ineq eq1}
	(g_m*g_m)(x)=\frac{1}{8\pi m}\mathrm{e}^{-m|x|},\qquad \mathrm{a.e.}\,x\in\mathbb{R}^3.
	\end{equation}
	Let $\mathscr{F}$ denote the Fourier transform on $L^2(\mathbb{R}^3)$ as defined in the proof of Lemma \ref{lemma for T_{m,V} eq if R+L^2}.
	Then it follows that
	\[
(2\pi)^{3/2}\left[\mathscr{F}(g_m*g_m)\right](k)=(2\pi)^{3}(\mathscr{F}g_m)(k)(\mathscr{F}g_m)(k)=\frac{1}{(|k|^2+m^2)^2},\qquad\mathrm{a.e.}\,k\in\mathbb{R}^3.
	\]
	On the other hand, a straightforward computation shows that
	\[
	\int_{\mathbb{R}^3}\mathrm{e}^{-m|x|}\mathrm{e}^{-\mathrm{i}x\cdot k}\,\mathrm{d}x=\frac{8\pi m}{(|k|^2+m^2)^2},\qquad\forall k\in\mathbb{R}^3.
	\]
	Hence, \eqref{lemma for T_{m,V} ineq eq1} holds.
	This completes the proof.
\end{proof}

\begin{proof}[Proof of Theorem \ref{T_{m,V} ineq}]
	We prove the theorem by contradiction.
	Assume that $\rho_{m,V_1}$ is equivalent to $\rho_{m,V_2}.$
	We first note that for each $j=1,2,$ the domain of $(T_{m,V_j})^2$ coincides with that of $(T_{m,0})^2,$ since $V_j$ belongs to $L^2+L_\varepsilon^\infty.$ 
	Thus, Proposition \ref{we suf cond for massive ops} implies that $(T_{m,V_j})^4\overset{\mathrm{w}}{\sim} (T_{m,0})^4.$
	Let $V:=V_1-V_2.$
	By Theorem \ref{m_1not=m_2} (3) and Corollary \ref{S^{-2}-T^{-2} is HS}, we see that the difference
	\[
	T_{m,V_2}^{-2}-T_{m,V_1}^{-2}=T_{m,V_2}^{-2}VT_{m,V_1}^{-2}
	\]
	is Hilbert--Schmidt.
	Since $(T_{m,V_j})^4\overset{\mathrm{w}}{\sim} (T_{m,0})^4$ for $j=1,2,$ the operator $T_{m,0}^{-2}VT_{m,0}^{-2}$ is also Hilbert--Schmidt.
	
	Define a potential $V_0:\mathbb{R}^3\to\mathbb{R}$ by setting $V_0(x):=V(x)$ for $|x|< R,$ and $V_0(x):=0$ otherwise.
	Let $V_\infty:=V-V_0.$ 
	Then, we have
	\[
	T_{m,0}^{-2}VT_{m,0}^{-2}=T_{m,0}^{-2}V_0T_{m,0}^{-2}+T_{m,0}^{-2}V_\infty T_{m,0}^{-2}.
	\]
	Since every function in $L^2+L_\varepsilon^\infty$ is locally square-integrable, $V_0$ belongs to $L^2(\mathbb{R}^3),$ and thus $T_{m,0}^{-2}V_0T_{m,0}^{-2}$ is Hilbert--Schmidt.
	This, together with the fact that $V_\infty$ is non-negative, implies that
	\[
	T_{m,0}^{-2}V_\infty T_{m,0}^{-2}=(V_\infty^{1/2} T_{m,0}^{-2})^*(V_\infty^{1/2} T_{m,0}^{-2})
	\]
    is Hilbert--Schmidt.
    Hence, 
    \[
    (V_\infty^{1/2} T_{m,0}^{-2})(V_\infty^{1/2} T_{m,0}^{-2})^*= \overline{V_\infty^{1/2}T_{m,0}^{-4}V_\infty^{1/2}}
    \]
    is also Hilbert--Schmidt.
	Note that the domain of $V_\infty^{1/2}T_{m,0}^{-4}V_\infty^{1/2}$ coincides with that of $V_\infty^{1/2},$ and hence $V_\infty^{1/2}T_{m,0}^{-4}V_\infty^{1/2}$ is densely defined.
	By Lemma \ref{lemma for T_{m,V} ineq}, we see that $\overline{V_\infty^{1/2}T_{m,0}^{-4}V_\infty^{1/2}}$ is an integral operator with kernel
	\[
	k_{m}(x,y):=\frac{1}{8\pi m}V_\infty(x)^{1/2}\mathrm{e}^{-m|x-y|}V_\infty(y)^{1/2},\qquad x,y\in\mathbb{R}^3.
	\]
	Therefore, $k_m$ must be square-integrable over $\mathbb{R}^3\times\mathbb{R}^3.$
	
	On the other hand, for any $x\in\mathbb{R}^3$ with $|x|\geq R+1,$ we have
	\begin{align*}
		&\int_{\mathbb{R}^3}\mathrm{e}^{-2m|x-y|}V_\infty(y)\,\mathrm{d}y
		\geq\int_{|x-y|\leq1}\mathrm{e}^{-2m|x-y|}V_\infty(y)\,\mathrm{d}y
		=\int_{|y|\leq1}\mathrm{e}^{-2m|y|}V_\infty(x-y)\,\mathrm{d}y\\
		&\qquad\geq\int_{|y|\leq1}\mathrm{e}^{-2m}V_\infty(x-y)\,\mathrm{d}y
		\geq\int_{|y|\leq1}\mathrm{e}^{-2m}f(|x|+1)\,\mathrm{d}y
		=\frac{4\pi}{3\mathrm{e}^{2m}}f(|x|+1).
	\end{align*}
	Thus, we obtain
	\begin{align*}
		&\frac{3\mathrm{e}^{2m}}{4\pi}\int_{\mathbb{R}^3}\int_{\mathbb{R}^3}V_\infty(x)\mathrm{e}^{-2m|x-y|}V_\infty(y)\,\mathrm{d}x\,\mathrm{d}y
		\geq\int_{|x|\geq R+1}V_\infty(x)f(|x|+1)\,\mathrm{d}x\\
		&\geq\int_{|x|\geq R+1}f(|x|+1)^2\,\mathrm{d}x
		=4\pi\int_{R+1}^\infty f(r+1)^2r^2\,\mathrm{d}r
		\geq4\pi\int_{R+1}^\infty f(r+1)^2\frac{(r+1)^2}{4}\,\mathrm{d}r.
	\end{align*}
	Since $f$ is non-negative and non-increasing, the right-hand side diverges.
	This contradicts the fact that $k_m$ is square-integrable.
	Therefore, $\rho_{m,V_1}$ is not equivalent to $\rho_{m,V_2}.$
\end{proof}

\begin{remark}
	It remains unclear to the authors whether Theorem \ref{T_{m,V} ineq}, which holds for potentials in $L^2 + L_\varepsilon^\infty$, also extends to potentials in $\mathcal{R} + L_\varepsilon^\infty$.
\end{remark}

We consider a potential $V_{\gamma,\sigma}$ defined by
\[
V_{\gamma,\sigma}(x):=\frac{\gamma}{|x|^{\sigma}},\qquad x\in\mathbb{R}^3\setminus\{0\},
\]
where $\gamma\in\mathbb{R}\setminus\{0\}$ and $0<\sigma<2$ are constants.
Note that $V_{\gamma,\sigma}$ belongs to $\mathcal{R}+L_\varepsilon^\infty.$

\begin{theorem}\label{coulomb type thm}
	Let $m>0,$ and let $V_1\in L^2+L_\varepsilon^\infty.$ 
	Suppose that $-\Delta+m^2+V_1$ belongs to $\mathcal{S\!A}_+(\mathscr{H}).$ 
	Let $\gamma\in\mathbb{R}\setminus\{0\}$ and $0<\sigma<2$ be such that $-\Delta+m^2+V_1+V_{\gamma,\sigma}$ belongs to $\mathcal{S\!A}_+(\mathscr{H}).$ 
	Set $V_2:=V_1+V_{\gamma,\sigma}.$
	Then $\rho_{m,V_1}$ is equivalent to $\rho_{m,V_2}$ if and only if $\sigma>3/2.$
\end{theorem}

\begin{proof}
	We first consider the case where $\sigma>3/2.$
	In this case, $V_{\gamma,\sigma}$ is in $\mathcal{R}+L^2$ and thus so is $V_1-V_2.$
	Thus, by Theorem \ref{T_{m,V} eq if R+L^2}, $\rho_{m,V_1}$ is equivalent to $\rho_{m,V_2}.$
	
	We next consider  the case where $\sigma<3/2$.
	In this case, $V_{\gamma,\sigma}$ is in $L^2+L_\varepsilon^\infty$ and thus so is $V_2.$
	If $\gamma<0$, then $V_1-V_2=f(|\cdot|),$ where $f(r):=|\gamma|/r^\sigma$ for $r>0.$
	Since $f$ is not square-integrable over $[1,\infty)$ with respect to the measure $r^2\,\mathrm{d}r,$ it follows from Theorem \ref{T_{m,V} ineq} that $\rho_{m,V_1}$ is not equivalent to $\rho_{m,V_2}.$
	On the other hand, if $\gamma>0$, then $V_2-V_1=f(|\cdot|)$, and thus $\rho_{m,V_1}$ is not equivalent to $\rho_{m,V_2}.$
	
	Finally, we consider the case where $\sigma=3/2.$
	Define a potential $V_0$ by setting $V_0(x):=V_{\gamma,\sigma}(x)$ for $|x|\geq1$ and $V_0(x):=0$ otherwise.
	Then, $V_{\gamma,\sigma}-V_0$ belongs to $L^{3/2}(\mathbb{R}^3)$, and thus it also belongs to $\mathcal{R}$.
	Moreover, $V_0$ belongs to $L^2+L_\varepsilon^\infty,$ and so does $V_1+V_0.$
	We show that $-\Delta+m^2+V_1+V_0$ is in $\mathcal{S\!A}_+(\mathscr{H}).$ 
	If $\gamma>0,$ then $V_0$ is non-negative.
	Since $-\Delta+m^2+V_1$ is in $\mathcal{S\!A}_+(\mathscr{H}),$ it follows that $-\Delta+m^2+V_1+V_0$ is also in $\mathcal{S\!A}_+(\mathscr{H}).$
	On the other hand, if $\gamma<0,$ then $V_0-V_{\gamma,\sigma}$ is non-negative.
	Since $-\Delta+m^2+V_2$ is in $\mathcal{S\!A}_+(\mathscr{H}),$ so is 
	\[
	-\Delta+m^2+V_1+V_0=(-\Delta+m^2+V_2)+(V_0-V_{\gamma,\sigma}).
	\]
	Thus, $\rho_{m,V_1+V_0}$ is defined, and by Theorem \ref{T_{m,V} eq if R+L^2}, it is equivalent to $\rho_{m,V_2}.$
	Moreover, by Theorem \ref{T_{m,V} ineq}, $\rho_{m,V_1+V_0}$ is not equivalent to $\rho_{m,V_1}.$
	Therefore, $\rho_{m,V_1}$ is not equivalent to $\rho_{m,V_2}.$
\end{proof}

\appendix

\section{Operator theory}

Let $\mathscr{K}$ be a complex Hilbert space.

\begin{lemma}\label{integral rep power p}
	Let $A$ be a non-negative self-adjoint operator acting in $\mathscr{K},$ and let $\alpha\in(0,1).$
	Then a vector $\psi\in\mathscr{K}$ is in $\mathrm{dom}(A^{\alpha/2})$ if and only if 
	\[
	\int_0^\infty \frac{\langle \psi, A(A+t)^{-1}\psi\rangle}{t^{1-\alpha}}\,\mathrm{d}t<\infty.
	\]
	Moreover, we have
	\[
	\|A^{\alpha/2}\psi\|^2
	=\frac{\sin{(\pi \alpha)}}{\pi}\int_0^\infty \frac{\langle\psi, A(A+t)^{-1}\psi\rangle}{t^{1-\alpha}}\,\mathrm{d}t,\qquad\forall\psi\in\mathrm{dom}(A^{\alpha/2}).
	\]
\end{lemma}

\begin{proof}
	See e.g., \cite[Proposition 5.16]{MR2953553}.
\end{proof}

\begin{lemma}\label{integral rep power p ver2}
	Let $A$ be a non-negative self-adjoint operator acting in $\mathscr{K},$ and let $\alpha\in(0,1).$
	Then we have
	\[
	\langle\psi,A^{\alpha}\phi\rangle
	=\frac{\sin{(\pi \alpha)}}{\pi}\int_0^\infty \frac{\langle\psi, A(A+t)^{-1}\phi\rangle}{t^{1-\alpha}}\,\mathrm{d}t,
	\qquad\forall\psi,\phi\in\mathrm{dom}(A^{\alpha}).
	\]
\end{lemma}

\begin{proof}
	The integral on the right-hand side converges absolutely by the polarization identity and Lemma \ref{integral rep power p}. 
	The equality also follows from the polarization identity and Lemma \ref{integral rep power p}. 
\end{proof}

\begin{lemma}\label{form op ine inverse}
	Let $A,B\in\mathcal{S\!A}_+(\mathscr{K})$ be arbitrary.
	If $A\preceq B,$ then $B^{-1}\preceq A^{-1}.$
\end{lemma}

\begin{proof}
	See e.g., \cite[Corollary 10.12]{MR2953553}.
\end{proof}

\begin{lemma}\label{Heinz ine}
	Let $A,B$ be non-negative self-adjoint operators acting in $\mathscr{K},$ and let $p\in(0,1].$
	If $A\preceq B,$ then $A^p\preceq B^p.$
\end{lemma}

\begin{proof}
	See e.g., \cite[Proposition 10.14]{MR2953553}.
\end{proof}

\begin{lemma}\label{gen int formula1}
	Let $A$ be a non-negative self-adjoint operator acting in $\mathscr{K},$ and let $\alpha\in[0,1).$
	Then, $A^{(1-\alpha)/2}(A+t)^{-1}$ is an everywhere defined bounded operator with
	\begin{equation}\label{eq;lem_gen int formula1}
	\int_{0}^\infty\|A^{(1-\alpha)/2}(A+t)^{-1}\psi\|^2t^\alpha\,\mathrm{d}t=C^{(\alpha)}\|\psi\|^2,\qquad\forall\psi\in\mathscr{K},
	\end{equation}
	where
	\[
	C^{(\alpha)}:=\int_{0}^\infty\frac{s^\alpha}{(1+s)^2}\,\mathrm{d}s=\begin{dcases*}
		\frac{\pi\alpha}{\sin{(\pi\alpha)}}, & $\alpha\in(0,1)$,\\
		1, & $\alpha=0$.
	\end{dcases*}
	\]
\end{lemma}

\begin{proof}
	Since the range of $(A+t)^{-1}$ coincides with $\operatorname{dom}(A)$, the operator $A^{(1-\alpha)/2}(A+t)^{-1}$ is everywhere defined.
	Moreover, by the spectral theorem, $A^{(1-\alpha)/2}(A+t)^{-1}$ is bounded, since the function
	\[
	[0,\infty)\to\mathbb{R},\qquad\lambda\mapsto\frac{\lambda^{(1-\alpha)/2}}{\lambda+t}
	\]
	is bounded.
	
	We next show \eqref{eq;lem_gen int formula1}.
	Let $A=\int_0^\infty\lambda\,\mathrm{d}E_A(\lambda)$ be the spectral resolution of $A.$
	Then, for any $\psi\in\mathscr{K}$, we have
	\begin{align*}
		\int_{0}^\infty\|A^{(1-\alpha)/2}(A+t)^{-1}\psi\|^2t^\alpha\,\mathrm{d}t
		&=\int_{0}^\infty\int_{0}^\infty\frac{\lambda^{1-\alpha} t^\alpha}{(\lambda+t)^2}\,\mathrm{d}\|E_A(\lambda)\psi\|^2\,\mathrm{d}t\\
		&=\int_{0}^\infty\int_{0}^\infty\frac{\lambda^{1-\alpha} t^\alpha}{(\lambda+t)^2}\,\mathrm{d}t\,\mathrm{d}\|E_A(\lambda)\psi\|^2.
	\end{align*}
	By substituting $t=:\lambda s$, we get
	\begin{align*}
		=\int_{0}^\infty\int_{0}^\infty\frac{s^\alpha}{(1+s)^2}\,\mathrm{d}s\,\mathrm{d}\|E_A(\lambda)\psi\|^2
		=C^{(\alpha)}\int_{0}^\infty\mathrm{d}\|E_A(\lambda)\psi\|^2
		=C^{(\alpha)}\|\psi\|^2.
	\end{align*}
	This completes the proof.
\end{proof}

\begin{lemma}\label{lem;ess_spec}
	Let $A,B$ be self-adjoint operators acting in $\mathscr{K}$, and let $z\in\mathbb{C}$ be such that $z\not\in\sigma(A)\cup\sigma(B)$.
	If $(A-z)^{-1}-(B-z)^{-1}$ is compact, then $\sigma_\mathrm{ess}(A)=\sigma_\mathrm{ess}(B)$.
\end{lemma}

\begin{proof}
	See e.g., \cite[Theorem XIII.14]{MR493421} or \cite[Theorem 8.12]{MR2953553}.
\end{proof}

\begin{lemma}\label{HS norm is lower semi-continuous}
	The Hilbert--Schmidt norm $\|\cdot\|_\mathrm{HS}$ on $\mathcal{B}(\mathscr{K})$, where $\|A\|_\mathrm{HS}:=+\infty$ if $A\in\mathcal{B}(\mathscr{K})$ is not Hilbert--Schmidt, is lower semi-continuous with respect to the weak operator topology. 
\end{lemma}

\begin{proof}
	It suffices to prove that the square of the Hilbert--Schmidt norm is lower semi-continuous.
	Let $\{e_n\}_{n=1}^\infty$ be an orthonormal basis of $\mathscr{K}.$
	Then it holds that
	\[
	\|A\|_\mathrm{HS}^2=\sum_{n=1}^\infty\|Ae_n\|^2,\qquad\forall A\in\mathcal{B}(\mathscr{K}).
	\]
	Since the sum of non-negative lower semi-continuous functions is again lower semi-continuous, it is enough to show that for each $\psi\in\mathscr{K},$ the map
	\[
	\mathcal{B}(\mathscr{K})\to[0,+\infty),\qquad A\mapsto \|A\psi\|^2
	\]
	is lower semi-continuous with respect to the weak operator topology.
	Since the supremum of lower semi-continuous functions is again lower semi-continuous, it follows that for each $\psi\in\mathscr{K},$ the map
	\[
	\mathcal{B}(\mathscr{K})\to[0,+\infty),\qquad A\mapsto \|A\psi\|^2=\sup_{\|\phi\|=1}|\langle\phi,A\psi\rangle|^2
	\]
	is lower semi-continuous with respect to the weak operator topology.
	This finishes the proof.
\end{proof}

Recall that the symbol $\mathcal{C}_p(\mathscr{K})$ denotes the set of all Schatten $p$-class operators on $\mathscr{K}$ for each $p\in[1,\infty).$

\begin{lemma}\label{noncommutative Holder ineq}
	Let $p,q,r\in[1,\infty)$ satisfy $r^{-1}=p^{-1}+q^{-1}.$
	Then, for any $A\in\mathcal{C}_p(\mathscr{K})$ and $B\in\mathcal{C}_q(\mathscr{K}),$ we have $AB\in\mathcal{C}_r(\mathscr{K}).$
\end{lemma}

\begin{proof}
	See e.g., \cite[Theorem 2.3.10]{MR3075382}.
\end{proof}

\section*{Acknowledgements}
YM was supported by JSPS KAKENHI Grant Numbers JP23K25783 and JP24K06755.
IS was supported by JSPS KAKENHI Grant Number JP20K03628.
AS was supported by JSPS KAKENHI Grant Number JP23K03229.
	
\section*{Data Availability}
No data were used to support this study.

\section*{Conflicts of interest}
The authors declare that they have no conflicts of interest.

\bibliographystyle{plain}
\bibliography{References}

\end{document}